\documentclass[11pt]{article}
\usepackage{geometry}                
\usepackage{amssymb,amsmath,amsthm}
\usepackage{epstopdf}
\usepackage{wrapfig}

\usepackage{marvosym}
\usepackage{tikz}
\usepackage{verbatim}

\newcommand{\CC}{\mathbb C}
\newcommand{\RR}{\mathbb R}
\newcommand{\DD}{\mathbb D}

\newtheorem{lemma}{Lemma}
\newtheorem{cor}{Corollary}
\newtheorem{theorem}{Theorem}

\newtheorem*{thm-others}{Theorem}

\newtheoremstyle{thm-others}
  {}
  {}
  {}
  {}
  {bold}
  {}
  {}{}

\title{
{Discontinuity in the asymptotic behavior of planar orthogonal
polynomials under a perturbation of the Gaussian weight}}
\author{Seung-Yeop Lee, Meng Yang}

\begin{document}
\maketitle
\begin{flushleft} \line(1,0){500} \end{flushleft}
\begin{abstract}
We consider the orthogonal polynomials, $\{P_n(z)\}_{n=0,1,\cdots}$,
with respect to the measure
$$|z-a|^{2c} e^{-N|z|^2}dA(z)$$ supported over the whole complex
plane, where $a>0$, $N>0$ and $c>-1$. We look at the scaling limit
where $n$ and $N$ tend to infinity while keeping their ratio, $n/N$,
fixed. The support of the limiting zero distribution is given in
terms of certain ``limiting potential-theoretic skeleton'' of the
unit disk. We show that, as we vary $c$, both the skeleton and the
asymptotic distribution of the zeros behave discontinuously at
$c=0$. The smooth interpolation of the discontinuity is obtained by
the further scaling of $c=e^{-\eta N}$ in terms of the parameter
$\eta\in[0,\infty).$
\end{abstract}
\begin{flushleft} \line(1,0){500} \end{flushleft}
\section{Introduction}

    Consider the ensemble of $n$ point particles,
$\{z_j\}_{j=1}^n\subset\CC$, distributed according to the
probability measure given by
\begin{equation}\label{eq:gibbs} \frac{1}{{\cal Z}_n} \prod_{i<j}|z_i-z_j|^2\cdot\exp\bigg(-N \sum_{j=1}^n Q(z_j)\bigg) \cdot \prod_{j=1}^ndA(z_j)\end{equation}
where ${\cal Z}_n$ is the normalization constant, $N>0$ is a (large)
parameter, $Q:\CC\to\RR\cup\{\infty\}$ is called an {\it external
potential} and $dA$ is the standard Lebesgue measure on the plane.

The statistical behavior of the particles has been studied \cite{Ma
2011} for a large class of potentials in various contexts including
random normal matrices and two-dimensional Coulomb gas.  For
example, in the scaling limit where $n$ and $N$ tend to infinity
while $n/N$ is fixed, it is known \cite{HM 2013} that the counting
measure of the particles converges weakly,
\begin{equation}\nonumber {\mathbb E} \frac{1}{N}\sum_{j=1}^n \delta(z-z_j) \to  \frac{\Delta Q}{4\pi}\chi_K \end{equation}
where $\Delta Q=(\partial_x^2+\partial_y^2) Q$, $\chi_K$ is the
indicator function of the compact set $K\subset\CC$ that we will
call a {\it droplet} following \cite{HM 2013}, and the expectation
value is taken with respect to the measure in \eqref{eq:gibbs}.

A connection to orthogonal polynomials can be provided by Heine's
formula. It says that the averaged characteristic polynomial of the
$n$ particles is the (monic) orthogonal polynomial of degree $n$,
i.e.,
$$ P_n(z)=P_{n,N}(z)={\mathbb E} \prod_{j=1}^n(z-z_j)   $$
satisfies the orthogonality condition,
\begin{equation}\label{ortho}
\int_\mathbb{C}P_{n,N}(z)\overline{P_{m,N}(z)}e^{-N
Q(z)}dA(z)=h_{n,N}\delta_{nm} \quad(n,m=0,1,2,\ldots),
\end{equation}
where $h_{n,N}$ is a (positive) norming constant. From this
connection, one might wonder if the zero distribution of $P_n$ would
tend to the averaged distribution of the particles.  Though this is
 the case with the orthogonal polynomials on the real line
(that corresponds to the particles confined on the line), in the
cases of two-dimensional orthogonal polynomials so far studied
\cite{Ba 2015, Ba1 2015, ku106 2015, ku104 2015, ku103 2015, ku94
2015}, the limiting zero distribution is observed to be concentrated
on a small subset of the droplet, on some kind of
potential-theoretic {\it skeleton} of $K$.\footnote{In some cases,
the skeleton is also called ``mother body'' \cite{Gu 1998, Gu
1999}.}

A {\it skeleton of} $K$ will refer to a subset of (the polynomial
hull of) $K$ with zero area, such that there exists a measure that
is supported exactly on the skeleton and that generates the same
logarithmic potential in the exterior of (the polynomial hull of)
$K$ as the Lebesgue measure supported on $K$. One characteristic of
such skeleton is that it can be discontinuous under the continuous
variation of the droplet $K$. A simple example \cite{Gu 1998} comes
from the sequence of polygons converging to a disk. The skeleton of
the polygon, which is the set of rays connecting each vertex to the
center, does not converge to the skeleton of the disk, the single
point at the center. Such discontinuity can also occur, as we will
see, when the perturbed droplets have real analytic boundary.

In this paper we ask whether the zero distribution of $P_n$ also
exhibits the similar discontinuity under the variation of the
underlying droplet or, equivalently, under the variation of the
external potential. We consider the external potential given by
\begin{equation}\label{qz}
Q(z)=|z|^2+ \frac{2c}{N}\log \frac{1}{|z-a|},\quad c>-1,\quad a>0.
\end{equation}
When $N$ is large and $c\ll N$, this represents a small perturbation
of the Gaussian weight. It corresponds to the interacting Coulomb
particles with charge $+1$ for each, in the presence of an extra
particle with charge $+c$ at $a$.   By a simple rotation of the
plane, the above $Q$ covers the case with any nonzero $a\in\CC$.

We are interested in the scaling limit where $N$ and $n$ go to
infinity while the ratio, $n/N$, is fixed to a positive number.
Below we will set $N=n$ without losing generality since the
orthogonality \eqref{ortho} gives the relation
$$ P_{n,\, N}(z; a) = \left(\frac{n}{N}\right)^{n/2} P_{n,\, n}\left(\sqrt{\frac{N}{n}}z; \sqrt{\frac{N}{n}}a\right),$$
where $P_{n,\, N}(z;a)=P_{n,\, N}(z)$ stands for the orthogonal
polynomial with respect to the external potential given by
\eqref{qz}. Though we will mostly use $N$, we will keep $n$ whenever
the expression holds true for general $n\neq N$.
\subsection{Limiting skeleton}

The potential \eqref{qz} has been studied in \cite{Ba 2015} with the
notation: $ c_\text{there}= c_\text{here}/N$. Let us denote
$c_\text{there}$ by $\gamma$ such that
$$  \gamma = \frac{c}{N}~~ \text{and}~~ Q(z)=|z|^2+ 2\gamma\log \frac{1}{|z-a|}. $$

To state Theorem \ref{thm1} let us introduce $K_\gamma$,
$\mu_\gamma$ and ${\cal S}_\gamma$,
 and define $\mu$ and ${\cal S}$.

Let $K_\gamma\subset\CC$ be the compact set, called a {\it droplet},
so that
$$\mu^{(2\text{D})}_\gamma = \frac{1}{4\pi} {\chi}_{K_\gamma}$$ is the unique
probability measure that minimizes the energy functional,
$$ I[\mu]= \int Q \,d \mu  + \frac{1}{2}\iint \log\frac{1}{|z-w|} d\mu(z)d\mu(w).$$
Let ${\cal S}_\gamma={\rm supp}\,\mu_\gamma$ be the {\it skeleton
of} $K_\gamma$, that is, the compact subset of $\CC$ with zero area
such that the probability measure $\mu_\gamma$ generates the same
logarithmic potential as $\mu_\gamma^{(2\text{D})}$:
\begin{equation}\label{balayage} U^{\mu_\gamma}(z) =U^{\mu_\gamma^{(2\text{D})}}(z),
\qquad z\notin \text{(polynomial convex hull of $K_\gamma$)}.
\end{equation} Here we denote
$U^m(z)=-\int \log|z-w| \,dm(w)$ for a positive Borel measure $m$.
We note that such skeleton may not be unique in general.  We give
explicit definitions of ${\cal S}_\gamma$ and $\mu_\gamma$ in
Section \ref{sec2}.

We define the {\it limiting skeleton} ${\cal S}$ by
\begin{equation}\label{skele} {\cal S}=\left\{z\in\CC : ~ {\rm Re}(\log z-az) =
\log\beta-a\beta,~~{\rm Re}\, z\leq\beta\right\},
\end{equation}
where
$$\beta=\min\{a, 1/a\}.$$
From the equivalent representation of ${\cal S}$ in the real
coordinates by
\begin{equation*} {\cal S}=
\left\{(x,y)\in{\mathbb R}^2 : x^2+y^2=\beta^2 e^{2a(x-\beta)},~
x\leq \beta \right\},
\end{equation*}
it is a simple exercise to show that, ${\cal S}\subset{\rm
clos}\,{\mathbb D}$ is a simple closed curve that encloses the
origin and intersects $\beta.$ We will denote the interior and the
exterior of ${\cal S}$ by ${\rm Int}\,{\cal S}$ and ${\rm
Ext}\,{\cal S}$ respectively. See Figure \ref{pic1} for some
illustration of ${\cal S}$.

We define $\mu$ to be the probability measure supported on ${\cal
S}$ given by
\begin{equation}\label{mu}  d\mu(z) =\rho(z)d\ell(z)=\frac{1}{2\pi
}\Big|a-\frac{1}{z}\Big|d\ell(z),\quad z\in{\cal S},
\end{equation}
 where $d\ell$ is the arclength
measure of ${\cal S}$. Alternatively, the same measure can be
written in terms of holomorphic differential by $d\mu(z) =(2\pi
\mathrm{i})^{-1}(1/z-a)\,dz.$

\begin{theorem}\label{thm1}
As $\gamma\to 0 $ we have the convergences,
 $$K_\gamma\to{\rm clos}\,\DD,\quad \mu_\gamma\to\mu,\quad {\cal S}_\gamma\to {\cal S},$$ in the appropriate senses (i.e., respectively in
 Hausdorff metric, in weak-$\ast$, and in Hausdorff metric).
\end{theorem}

The proof is in Section \ref{sec2}.\vspace{0.2cm}

\noindent{\bf Remark 1.} In both examples, the one by Gustafsson
\cite{Gu 1998} and the one from the above theorem -- the
discontinuity occurs when the droplet becomes a disk. It is an
interesting question whether the discontinuity occurs with other
shapes than disk. We think that, at least for an algebraic potential
where the exterior of the droplet is a quadrature domain, the
discontinuity happens only with the disk. This is for the simple
reason that the disk is the only quadrature domain where the
exterior domain is also a quadrature domain.

\subsection{Strong asymptotics of $P_N$ and the location of zeros}

Let us define
\begin{equation}\label{phi0}
\begin{split}
 &\phi_{A}(z) = a(z-\beta)-\log\frac{z}{\beta},\\
&\phi(z)=\left\{\begin{array}{l}\displaystyle \phi_{A}(z),\qquad
z\in {{\rm Ext}\,{\cal S}},\\\displaystyle -\phi_{A}(z),\quad\, z\in
{{\rm Int}\,{\cal S}}.
\end{array}\right.
\end{split}
\end{equation}
Note that ${\rm Re}\,\phi\equiv 0$ on ${\cal S}$.

Let $U$ be a certain neighborhood of ${\cal S}\setminus \{\beta\}$
where ${\rm Re}\,\phi\leq 0$. See Figure \ref{lens} and the
paragraph below Lemma \ref{lemma1} for more details. Let $D_\beta$
be a disk neighborhood of $\beta$ with a fixed radius such that the
map $\zeta: D_\beta\to \mathbb{C}$ given below is univalent.
\begin{equation}\label{zetamap}
\zeta(z)=\left\{\begin{array}{l}\displaystyle \sqrt{
2N\phi_{A}(z)}=a\sqrt{N}(z-\beta)(1+\mathcal {O}(z-\beta))
\qquad\,\text{for}\quad a>1,\\\displaystyle
-N\phi_{A}(z)=\frac{1-a^2}{a}N(z-\beta)\left(1+\mathcal
{O}(z-\beta)\right) \quad\text{for}\quad a<1.
\end{array}\right.
\end{equation}
\begin{theorem}\label{thumm}
For $a>1$\,and for any fixed nonzero $c>-1$, we have
\begin{equation*}
P_N(z) =\left\{
\begin{array}{ll}\displaystyle
z^N\left(\frac{z}{z-\beta}\right)^{c}\left(1+\mathcal
{O}\left(\frac{1}{N}\right)\right),
 & z\in {\rm Ext}\,{\cal S}\setminus (U\cup D_{\beta}),\vspace{0.4cm}\\\displaystyle
-\frac{\beta^N\sqrt{2\pi}(a^2-1)^c}{N^{1/2-c}a\Gamma(c)}\frac{e^{Na(z-\beta)}}{z-\beta}
\left(\frac{z-\beta}{z-a}\right)^{c}
\left(1+\mathcal{O}\left(\frac{1}{\sqrt N}\right)\right),
 & z\in {\rm Int}\,{\cal S}\setminus (U\cup D_{\beta}),\vspace{0.4cm}\\\displaystyle
\begin{split}&z^N\left(\frac{z}{z-\beta}\right)^{c}\left(1+\mathcal
{O}\left(\frac{1}{N}\right)\right) \vspace{0.2cm}\\\displaystyle
&\qquad
-\frac{\beta^N\sqrt{2\pi}(a^2-1)^c}{N^{1/2-c}a\Gamma(c)}\frac{e^{Na(z-\beta)}}{z-\beta}
\left(\frac{z-\beta}{z-a}\right)^{c} \left(1+\mathcal
{O}\left(\frac{1}{\sqrt N}\right)\right), \end{split} & z\in
U\setminus D_{\beta},\vspace{0.4cm}
 \\\displaystyle
z^N\left(\left(\frac{z\zeta(z)}{z-\beta}\right)^{c}e^{\frac{\zeta^2(z)}{4}}D_{-c}(\zeta(z))
+\mathcal{O}\left(\frac{1}{\sqrt N}\right)\right), & z\in D_{\beta}.
\end{array}
\right.
\end{equation*}
Here $D_{{-c}}$ be the Parabolic cylinder function or Weber function
and is defined by \cite{Ha 2010}
\begin{equation}\label{para!}
D_{{-c}}(\zeta):=\frac{e^{\frac{\zeta^2}{4}}}{\mathrm{i}\sqrt{2\pi}}\int_{\epsilon-\mathrm{i}\infty}^{\epsilon+\mathrm{i}\infty}e^{-\zeta
s+\frac{s^2}{2}}s^{-c}ds,\quad \epsilon>0.
\end{equation}
\end{theorem}

\begin{theorem}\label{thum33}
For $a<1$ and for any fixed nonzero $c>-1$, we have
\begin{equation}\label{pnp}
P_N(z)=\left\{
\begin{array}{lll}
\displaystyle z^N\left(\frac{z}{z-a}\right)^{c}\left(1+\mathcal
{O}\left(\frac{1}{ N^\infty}\right)\right),  &z\in {\rm Ext}\,{\cal
S}\setminus (U\cup D_{\beta}),\vspace{0.4cm}\\\displaystyle
-\frac{a^{1+N}(1-a^2)^{c-1}}{N^{1-c}\Gamma(c)}\frac{e^{Na(z-a)}}{z-a}\left(1+\mathcal
{O}\left(\frac{1}{N}\right)\right), &z\in {\rm Int}\,{\cal
S}\setminus (U\cup D_{\beta}),\vspace{0.4cm}\\\displaystyle
\begin{split}
& z^N\left(\frac{z}{z-a}\right)^{c}\left(1+\mathcal
{O}\left(\frac{1}{
N^\infty}\right)\right)\vspace{0.2cm}\\
&\qquad-\frac{a^{1+N}(1-a^2)^{c-1}}{N^{1-c}\Gamma(c)}\frac{e^{Na(z-a)}}{z-a}\left(1+\mathcal
{O}\left(\frac{1}{N}\right)\right),\end{split} &z\in U\setminus
{D}_{\beta},\vspace{0.4cm}\\\displaystyle
\begin{split}
& z^N\bigg(\left(\frac{z}{z-a}\right)^{c}\left(1+\mathcal
{O}\left({\frac{1}{ N^\infty}}\right)\right) \vspace{.2cm}\\
&\qquad\displaystyle
-\left(\frac{z\zeta(z)}{z-a}\right)^{c}\frac{1}{e^{\zeta(z)}}\left({
\hat{f}}(\zeta(z))+\mathcal
{O}\left(\frac{1}{N}\right)\right)\bigg),\end{split} &
z\in{D}_{\beta}.

\end{array}
\right.
\end{equation}
Here
\begin{equation}\nonumber
\hat{f}(\zeta)=\frac{-1}{2\mathrm{i}\pi}\int_{\cal
L}\frac{e^{s}}{s^{c}(s-\zeta)}ds,
\end{equation}
where the contour ${\cal L}$ begins at $-\infty$, circles the origin
once in the counterclockwise direction, and returns to $-\infty$.
The error bound ${\cal O}(1/N^\infty)$ means ${\cal O}(1/N^k)$ for
arbitrary integer $k$.
\end{theorem}
One can check that the branch cut discontinuity of $(z/(z-a))^c$ in
the last equation of \eqref{pnp} is canceled by the discontinuity of
$\hat{f}$ so that the asymptotic expression of $P_N$ in $D_\beta$ is
analytic.

From Theorem \ref{thumm} and \ref{thum33}, one can notice that the
zeros of $P_N$ can appear when the two terms in the asymptotic
expressions of $P_N$ in $U\setminus D_\beta$ cancel each other and
hence must have the same order in $N$. Such cancellation may be
expressed in terms of $\phi_{A}$ as we presently explain  below.
\begin{align}\nonumber
\left(\frac{z}{z-\beta}\right)^{c} &= e^{N\phi_{A}(z)}
\left(\frac{z-\beta}{z-a}\right)^{c}\frac{\sqrt{2\pi}(a^2-1)^c}{a\Gamma(c)N^{\frac{1}{2}-c}(z-\beta)},&\text{for}\quad a>1,\\
\nonumber\left(\frac{z}{z-a}\right)^{c} &= e^{
N\phi_{A}(z)}\frac{a(1-a^2)^{c-1}}{N^{1-c}\Gamma(c)(z-a)},&\text{for}\quad
a<1.
\end{align}
Taking the logarithm of the absolute values on both sides and after
simple calculations, we get
\begin{align}\label{zerolocation1}
- {\rm Re}\,\phi_{A}(z) &=\bigg(c-\frac{1}{2}\bigg)\frac{\log
N}{N}-\frac{\log\Gamma(c)}{N}+\frac{1}{N}\log{\left|\left(\frac{z-\beta}{z-a}\right)^{c}\frac{\sqrt{2\pi}(a^2-1)^c}{a(z-\beta)^{1-c}z^c}\right|},\,\,&a>1,\\\displaystyle
\label{zerolocation2}- {\rm Re}\, \phi_{A}(z)&=\frac{(c-1)\log
N}{N}-\frac{\log\Gamma(c)}{N}+\frac{1}{N}\log{\bigg|\frac{a(1-a^2)^{c-1}}{(z-a)^{1-c}z^c}\bigg|},\,\,&a<1.
\end{align}
As we will show in Lemma \ref{lemma1}, ${\rm Re}\,\phi_{A}$ is
positive (resp. negative) in $U\cap{\rm Int}\,{\cal S}$ (resp. in
$U\cap{\rm Ext}\,{\cal S}$). For $a>1,$ since the dominant term in
the right hand side of \eqref{zerolocation1} is
$\left(c-\frac{1}{2}\right)\frac{\log N}{N}$, the zeros will
approach ${\cal S}$ from ${\rm Ext}\,{\cal S}$ for $c>\frac{1}{2}$
and from ${\rm Int}\,{\cal S}$ for $c<\frac{1}{2}$. For $a<1$, since
the dominant term in the right hand side of \eqref{zerolocation2} is
$\left(c-1\right)\frac{\log N}{N}$, the zeros will approach ${\cal
S}$ from ${\rm Ext}\,{\cal S}$ for $c>1$ and from ${\rm Int}\,{\cal
S}$ for $c<1.$  See Figure \ref{pic1}. We also remark, without
proof, that the limiting distribution of the zeros is given by $\mu$
which is explicitly given in \eqref{mu}. This can be proven, for
example, using the method in \cite{Tot 1997} (Chapter III) and
\cite{Mha 1991} (Theorem 2.3).

\begin{figure}
\begin{center}
\includegraphics[width=0.495\textwidth]{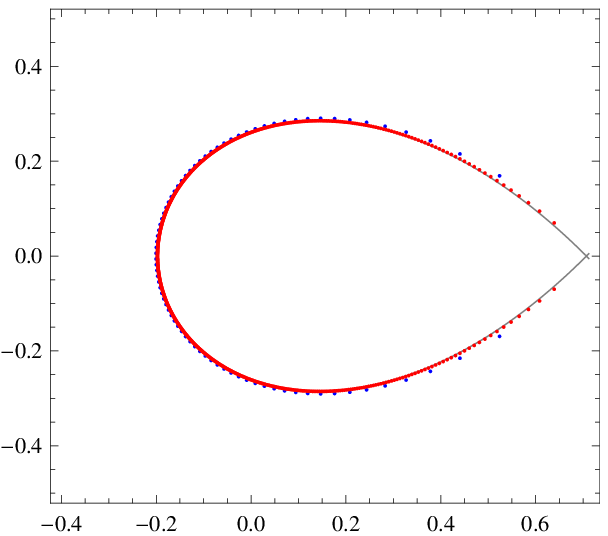}
\includegraphics[width=0.49\textwidth]{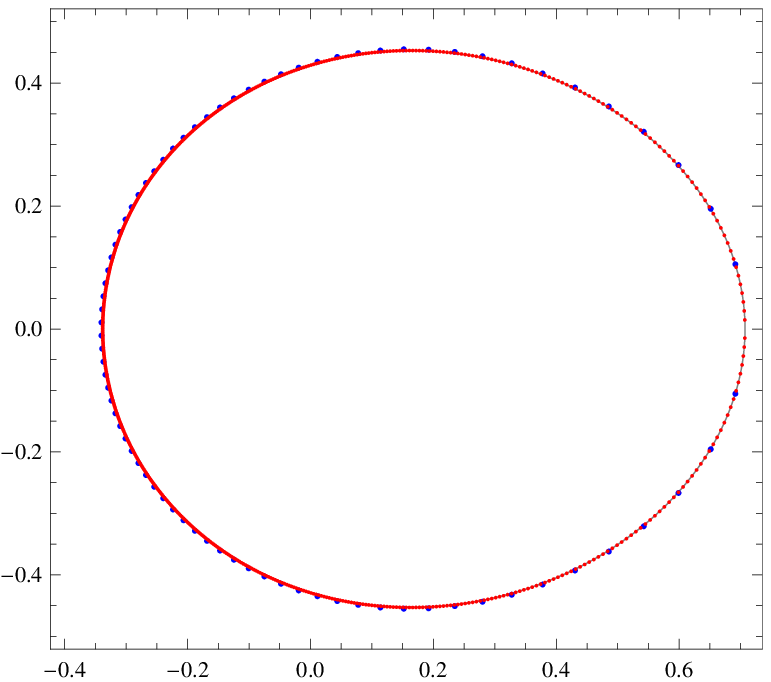}
\end{center}
\caption{The zeros of orthogonal polynomials with degrees $80$
(blue) and $600$ (red) for $c=1$. The left is for $a=\sqrt{2}$ and
the right is for $a=1/\sqrt{2}$.  In both cases, zeros are close to
the curves representing ${\cal S}$.} \label{pic1}
\end{figure}

\begin{figure}
\begin{center}
\includegraphics[width=0.465\textwidth]{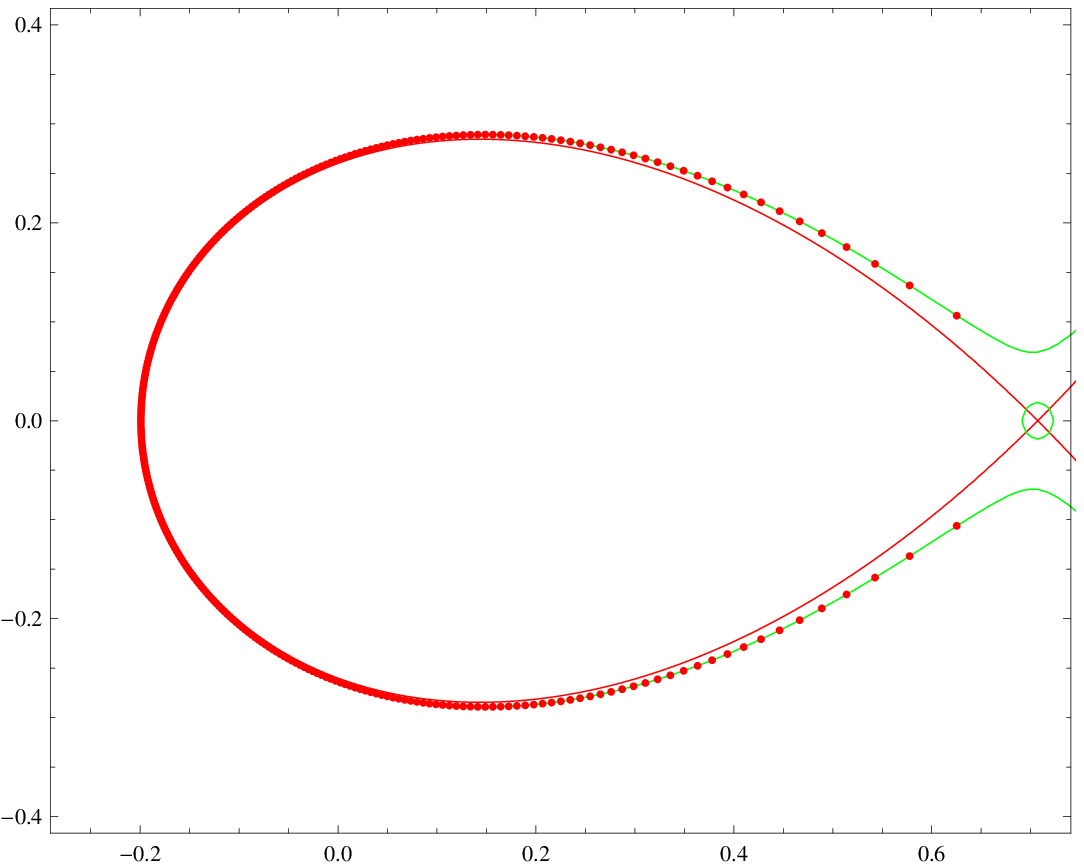}
\includegraphics[width=0.4\textwidth]{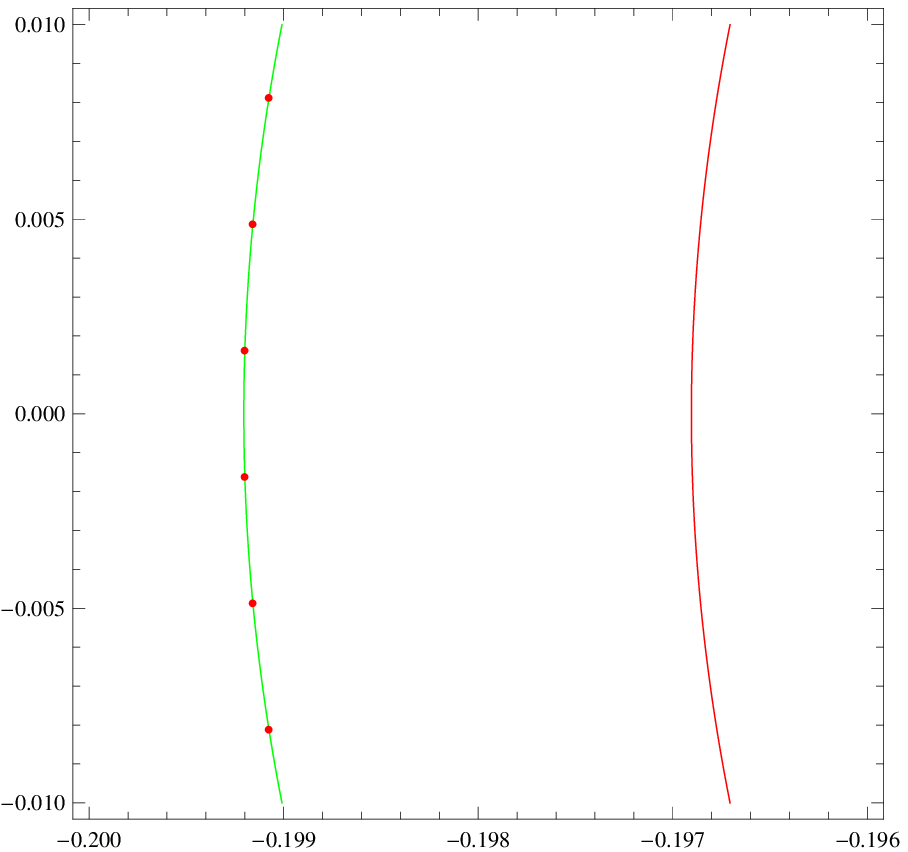}
\end{center}
 \caption{Zeros of orthogonal polynomials when $a=\sqrt 2$, $c=1$ and $N=300$. The red line is ${\cal S}$ and the green line is the solution set of \eqref{zerolocation1}. The right
figure is the enlarged view of the left figure. } \label{fig1}
\end{figure}

\begin{figure}
\begin{center}
\includegraphics[width=0.48\textwidth]{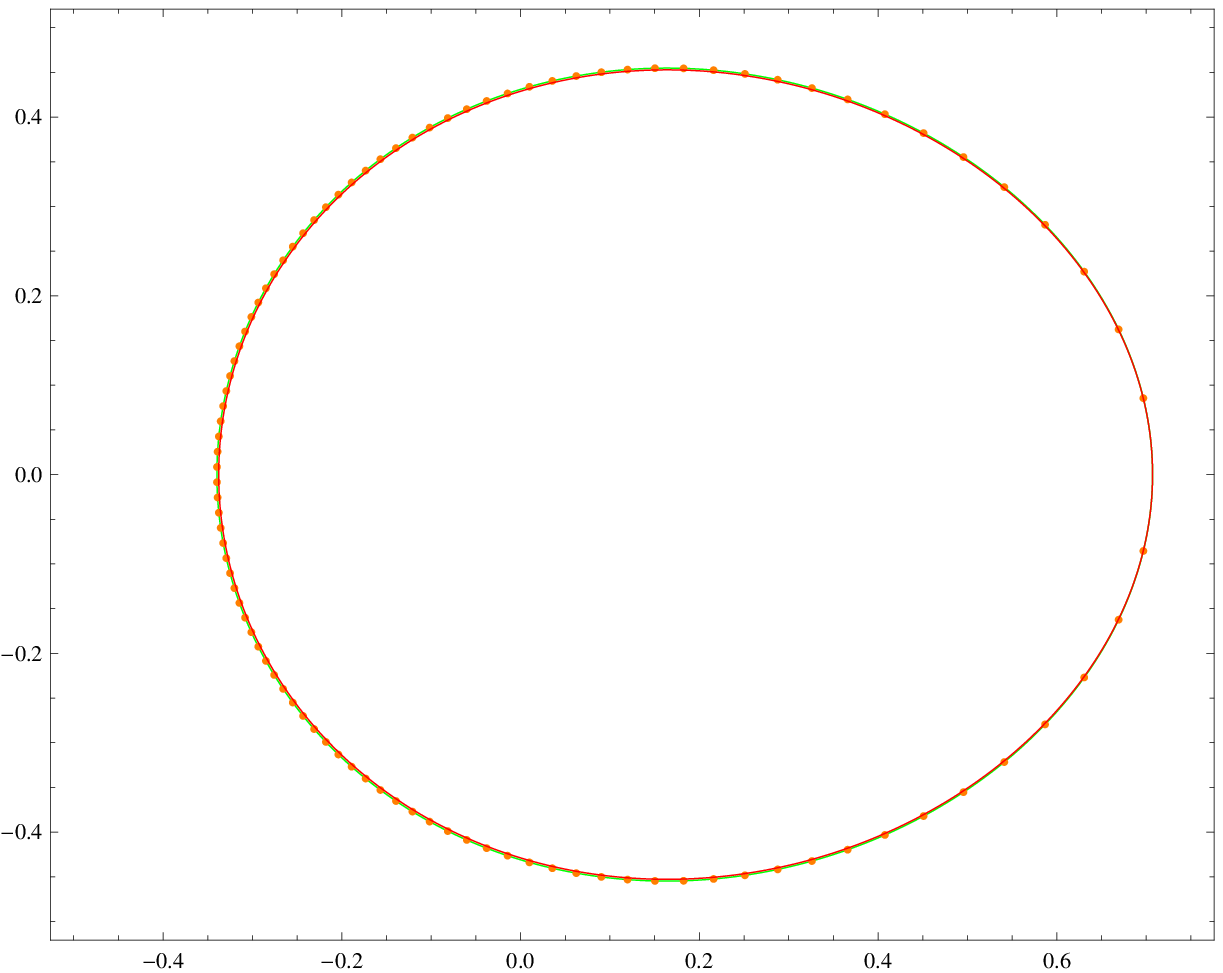}
\includegraphics[width=0.4\textwidth]{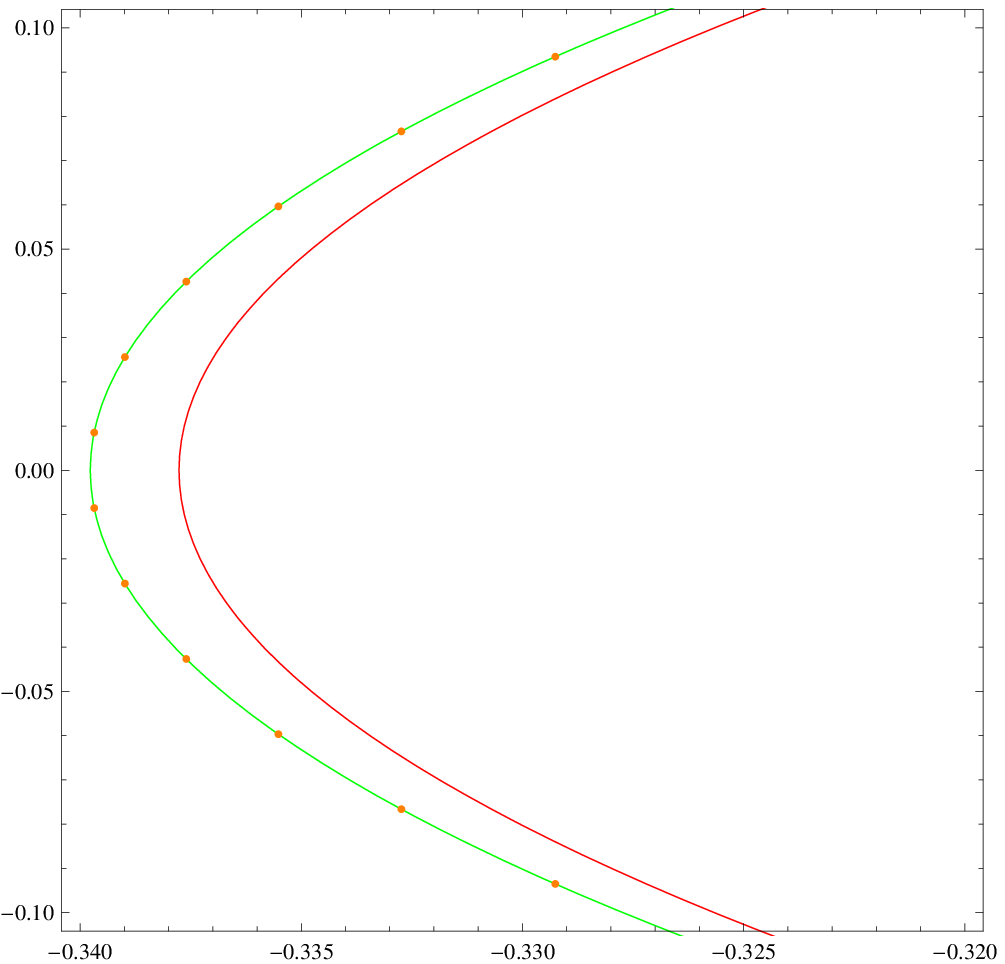}
\end{center}
\caption{When $a=1/\sqrt 2$, $c=1$ and $N=100.$ The red line is
${\cal S}$ and the green line is the solution set of
\eqref{zerolocation2}. The right figure is the enlarged view of the
left figure.} \label{fig3}
\end{figure}

We remark that the case $-1<c<0$ is essentially treated in \cite{Ba1
2015}.  We note that the limiting locus of zeros remains the same
for both the positive and negative $c$ (which seems unexpected
according to Remark 1.2 in \cite{Ba1 2015}).  It turns out that, as
the value of $c$ gets bigger, we need higher order corrections in
the Riemann-Hilbert analysis. To obtain the result that works for an
arbitrary value of $c$, therefore, we need an arbitrary order
correction in the Riemann-Hilbert analysis. This is done in Section
\ref{rs} using the method developed in \cite{BL 2008}.

We found that the limiting support of the zeros does not depend on
$c$. Even for $c$ algebraically decaying in $N$ (e.g., $c=
N^{-1000}$) the limiting support of the zeros converges to ${\cal
S}$. However, when $c$ decays exponentially in $N$, say $c=e^{-\eta
N},$  the right hand sides of both \eqref{zerolocation1} and
\eqref{zerolocation2} converge to
$$-\eta = -\lim_{N\to\infty} \frac{\log\Gamma(e^{-\eta N})}{N},\quad \eta>0$$
and the zeros approach the curve in ${\rm Int}\,{\cal S}$ given by
the equation
\begin{equation}\label{etacurve}
{\rm Re}\,\phi_{A}(z)=\eta.
\end{equation}
A similar ``sensitive behavior of zeros under a parameter" has been
observed in \cite{ku105 2015}.

It is simple to observe that the family of curves given by
\eqref{etacurve} for $0\leq\eta<\infty$ continuously interpolates
between the curve ${\cal S}$ and the origin. In Figure \ref{pic2},
we show the curves satisfying \eqref{etacurve} for $\eta=0.2$ and
$\eta=0.4$, with the corresponding zeros.

To establish the behavior of zeros for {\em scaling $c$}, however,
Theorem \ref{thumm} and \ref{thum33} are not enough as the error
bounds in the theorems are for {\em fixed $c$}.  For $c$ {\em that
scales to zero with $N$} we will prove Theorem \ref{thma1} and
\ref{thm6} where the error bounds are {\em uniform in }
$c$.\vspace{0.2cm}

\noindent{\bf Remark 2.} A simple way to understand the phenomenon
is to recall the well--known instability of roots of polynomials,
for example, the zeros of $P_n(z) = z^n + a/n^k$ still tend to the
uniform distribution on the unit circle as $n\to \infty$ (for any
fixed positive $k$) although the polynomial is a ${\cal O}(n^{-k})$
perturbation of the monomial. This simple toy example already shows
that a perturbation that interpolates between the two behaviors
would require to have $a = e^{-n\eta}$. In this perspective it is
not unexpected to see the exponentially small perturbations of the
orthogonality measure in order to interpolate the behaviors.
\vspace{0.2cm}

\begin{figure}
\begin{center}
\includegraphics[width=0.495\textwidth]{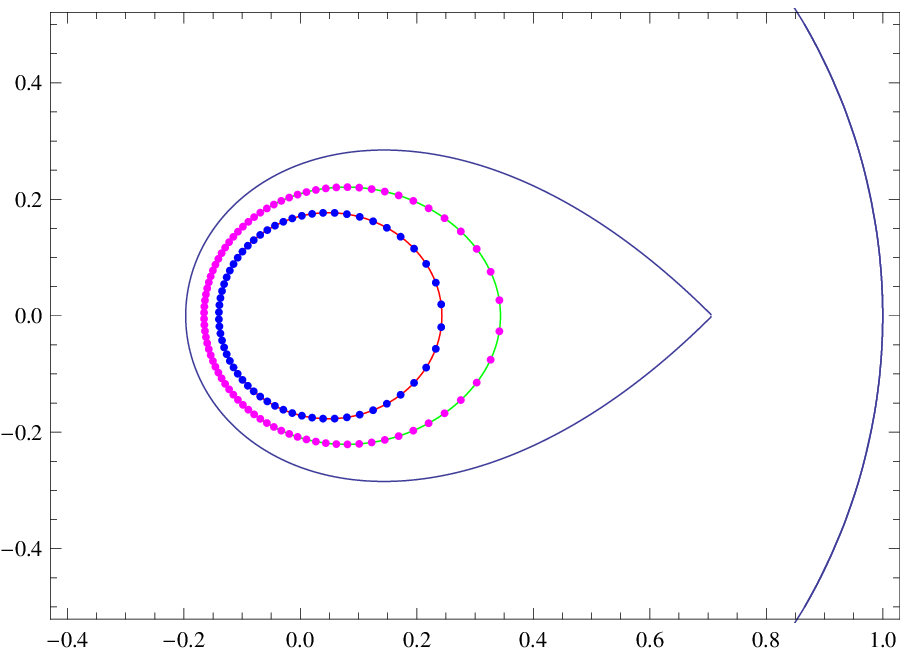}
\includegraphics[width=0.495\textwidth]{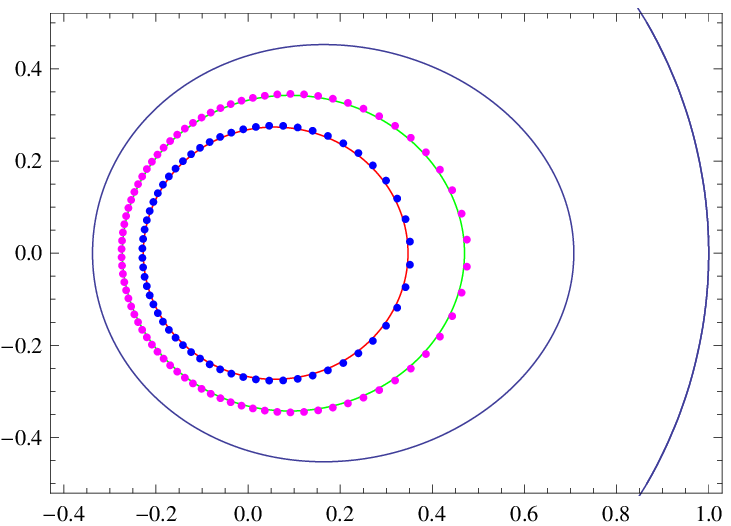}
\end{center}
\caption{The zeros of orthogonal polynomials with degrees $60$
(blue) and $80$ (magenta) for $c=e^{-\eta n}$, where $\eta=0.4$
(blue) and $\eta=0.2$ (magenta). The left is for $a=\sqrt{2}$ and
the right is for $a=1/\sqrt{2}$. In both cases, zeros seem to
converge to the curves given by \eqref{etacurve} of the
corresponding values.} \label{pic2}
\end{figure}

\noindent{\bf Remark 3.} A main message of the paper is that the
asymptotic zero locus can be quite sensitive to the small
perturbation of the underlying measure. In Figure \ref{pic4} we give
another numerical plot that supports such statement. The example
considers the orthogonal polynomials with the cutoff. Though the
cutoff may be considered as a ``small perturbation'' to the
underlying Coulomb particle system, it seems to affect the
polynomial significantly.\vspace{0.2cm}

In the next section we prove Theorem 1 about the limiting skeleton.
In section 3 we prove the asympototic result for $a>1$ and $c$ near
0. In section 4 we prove the similar result for an arbitrary $c$. In
section 5 we prove the asympototic result mostly following the
arguments from the previous two sections. In the last section, we
argue that the similar method will give the result for the critical
case of $a\approx1$, by showing that the local parametrix satisfies
the Riemann-Hilbert problem for Painlev\'e IV
equation.\vspace{0.2cm}

\noindent{\bf Acknowledgement.} The first author was supported by
Simons Collaboration Grants for Mathematicians. We thank the referee
for the insightful remarks -- Remark 1 and 2 were prompted by the
referee.

 \begin{figure}
 \centering
 \includegraphics[width=0.50\textwidth]{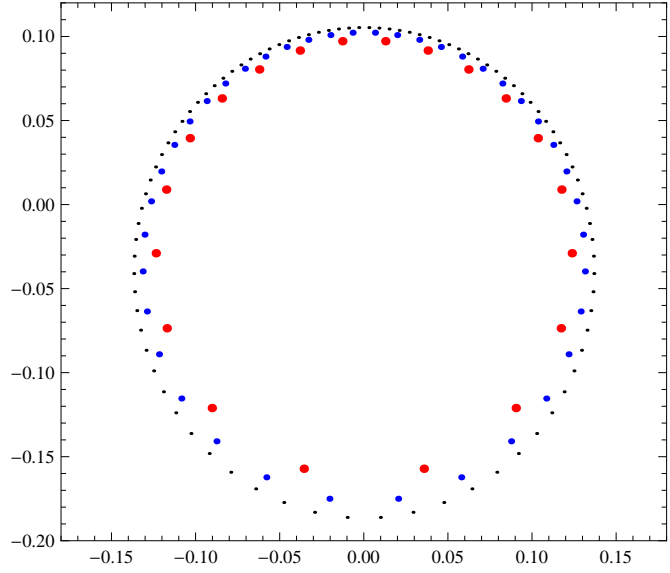}
 \caption{The zeros of orthogonal polynomials with degrees
 from $\{ 20,
40,  90\}$ and with the orthogonality measure given by $\chi_K
\exp(-n|z|^2) dA(z) $ where $K=(-\infty,
+\infty)\times[-{3\mathrm{i}}/{2},+\mathrm{i}\infty)\subset\CC$. The
plot suggests that the limiting support of zeros is not the
origin.}\label{pic4}
 \end{figure}

\section{The proof of Theorem \ref{thm1}}\label{sec2}

For the convenience of the readers we reproduce the useful
definitions in \cite{Ba 2015}.

For $a<1$ for a sufficiently small $\gamma$ we define
\begin{equation}{\label{ap1}}
 K_\gamma=\overline{D\big(0,\sqrt{1+\gamma}\big)}\setminus
D(a,\sqrt{\gamma}),
\end{equation}
 where $D(a,r)$ stands for the disc with radius
$r$ centered at $a$.

For $a<1$ we define ${\cal S}_\gamma$ to be the simple closed curve
enclosing $[0,a]$ and intersecting
$$\beta_\gamma=\frac{a^2+1-\sqrt{(1-a^2)^2-4a^2\gamma}}{2a}>a,$$
such that the quadratic differential $y_\gamma(z)^2dz^2$ is real and
negative on $\mathcal {S}_\gamma$ where
 $$
y_\gamma(z):=(-1)^{\chi_{\text{Int}\mathcal
{S}_\gamma}}\left[a+\frac{\gamma}{z-a}-\frac{1+\gamma}{z}\right].
$$
Here, we denote the interior of the simple closed curve ${\cal
S}_\gamma$ by $\text{Int}\,{\cal S}_\gamma$. We recall that $\chi$
is the indicator function.

For $a\geq1,$ the set $K_\gamma$ is defined to be the closure of the
interior of the real analytic Jordan curve given by the image of the
unit circle under $f_\gamma$ given by
$$f_\gamma(\nu)=\rho\nu-\frac{\kappa
}{\nu-\alpha}-\frac{\kappa}{\alpha},$$ whose parameters $\rho>0,
\kappa\geq 0,$ and $0<\alpha\leq 1/a$ are given in terms of $a$ and
$\gamma$ below. First, $\rho$ and $\kappa$ are given by
$$\rho=\frac{1+a^2\alpha^2}{2a\alpha},\quad \kappa=\frac{(1-\alpha^2)(1-a^2\alpha^2)}{2a\alpha}.$$
The parameter $\alpha$ is given by the unique solution of
$P_\gamma(\alpha^2)=0$ such that $0<\alpha\leq 1/a$ where
$$P_\gamma(X):=X^3-\left(\frac{a^2+4\gamma+2}{2a^2}\right)X^2+\frac{1}{2a^4}.$$
The uniqueness is easily seen by $P_\gamma(0)>0$ and
$P_\gamma(1/a^2)= -2\gamma/a^6<0$.  We note that, as $\gamma$ goes
to zero,
 $\alpha$ goes to $1/a$, $\kappa$ goes to zero and $\rho$ goes to 1.

For $a\geq 1$ we define $\mathcal {S}_\gamma$ to be the smooth arc
with the endpoints at
$$\beta_\gamma:=\alpha\rho-\frac{\kappa}{\alpha}+2\mathrm{i}\sqrt{\kappa\rho}\quad \text{and}\quad \overline{\beta_\gamma}$$
such that the quadratic differential $y_\gamma(z)^2dz^2$ is real and
negative on $\mathcal {S}_\gamma$ where
$$y_\gamma(z):=\frac{a(z-b_\gamma)\sqrt{(z-\beta_\gamma)(z-\overline{\beta_\gamma})}}{z(z-a)},\quad b_\gamma=\frac{\rho}{\alpha}.\footnote{In \cite{Ba 2015} $b_\gamma$ is written as $\alpha/\rho$ by mistake.   } $$

For all values of $a$, we define the probability measure
$\mu_\gamma$ supported on ${\cal S}_\gamma$ by
$$d\mu_\gamma=\frac{1}{2\pi}|y_\gamma(z)|d\ell_\gamma,$$ where $d\ell_\gamma$ is the arclength
measure of ${\cal S}_\gamma.$

For all values of $a$, we define $\phi_\gamma$ by
$$\phi_\gamma(z)=\int_{\beta_\gamma}^z y_\gamma(s)\,ds,$$
where the integration contour lies in the simply connected domain
$\mathbb{C}\setminus([0,\infty)\cup
[\beta_\gamma,\overline{\beta_\gamma}])$, where
$[\beta_\gamma,\overline{\beta_\gamma}]$ stands for the vertical
line segment connecting $\beta_\gamma$ and $\overline{\beta_\gamma}$
(for $a\geq1$, $[\beta_\gamma,\overline{\beta_\gamma}]$ is a point
on $\mathbb{R}^+$). One can consider $\phi_\gamma$ to be defined
{\em over the whole complex plane} by analytic continuation over
$[0,\infty)\cup [\beta_\gamma,\overline{\beta_\gamma}]$ {\em
consistently} for all $\gamma$.

\begin{lemma}\label{lem1}
   As $\gamma$ goes to $0$, $\phi_\gamma$ converges to
$\phi_0:=\phi_{\gamma=0}$ uniformly over a compact subset in
$\mathbb{C}\setminus\{0,a\}$.
\end{lemma}
\begin{proof}
It is simple to check that, as $\gamma$ goes to zero, $\beta_\gamma$
converges to $\beta$ and $b_\gamma$ converges to $a$.  Therefore
$y_\gamma(z)$ converges to $y_{\gamma=0}(z)$, by choosing the branch
cut of $y_\gamma$ at $[\beta_\gamma, \overline{\beta_\gamma}]$ that
converges to $\beta$.  This convergence is uniform away from the
singularities of $y_\gamma$ at $0$ and $a$.
\end{proof}

\begin{lemma}\label{lemma20} Let $I=\{ {\rm i}t:-2\pi \leq t\leq 0 \}$. The mapping $\phi_\gamma:{\cal S}_\gamma\setminus\{\beta_\gamma,\overline{\beta_\gamma}\}\to I\setminus\{0,-{\rm i}2\pi\}$ is invertible.
\end{lemma}

\begin{proof}
We prove this for $a>1$ as the other case is similiar.  We get
$\phi_\gamma(\beta_\gamma)=0$ by definition. We have
$$\phi_\gamma(\overline{\beta_{\gamma}})=\int_{\beta_\gamma}^{\overline{\beta_{\gamma}}}
y_\gamma(s)\,ds=\frac{1}{2}\oint y_\gamma(s)\,ds,$$ where, in the
first integral, the integration contour can be taken along ${\cal
S}_\gamma$ and, in the second integral, the integration contour goes
{\em around} ${\cal S}_\gamma$ counterclockwise while the branch cut
of $y_\gamma$ is placed at ${\cal S}_\gamma$ (instead of at
$[\beta_\gamma, \overline{\beta_\gamma}]$). The latter integration
contour can be deformed into three clockwise contours around
$\infty, 0$ and $a$, which leads to
$$\phi_\gamma(\overline{\beta_{\gamma}})=-\frac{2\pi\mathrm{i}}{2}\left(\mathop{\textrm{Res}}_{z=\infty} y_\gamma(z)
+\mathop{\textrm{Res}}_{z=0} y_\gamma(z)+\mathop{\textrm{Res}}_{z=a}
y_\gamma(z)\right).$$ By Lemma 2.19 in \cite{Ba 2015}, we have
$\textrm{Res}_{z=\infty} y_\gamma(z)=1, \textrm{Res}_{z=0}
y_\gamma(z)=1+\gamma,$ and $\textrm{Res}_{z=a} y_\gamma(z)=-\gamma$
and, therefore, we have
$\phi_\gamma(\overline{\beta_{\gamma}})=-2\pi\mathrm{i}.$  Since
$\phi_\gamma$ is continuous on $S_\gamma$ (here we again place the
branch cut of $y_\gamma$ at $[\beta_\gamma,
\overline{\beta_\gamma}]$) we have $I\subset\phi_\gamma({\cal
S}_\gamma)$. Since $\phi_\gamma$ has no critical point in ${\cal
S}_\gamma$ except at the endpoints, $\phi_\gamma$ is 1-to-1 and
$I=\phi_\gamma({\cal S}_\gamma)$.
\end{proof}

\begin{lemma}\label{lemma3} Let $\{K_j\subset\CC\}_{j=1}^\infty$ be bounded a sequence of compact sets such that $K_\infty$, the set of limit points of $\{K_j\}_{j=1}^\infty$, is also compact. If $K_j$'s are all connected such that $b_j\in K_j$ and $\lim_{j\to\infty} b_j= b_\infty$ then $K_\infty$ is connected to $b_\infty$.
\end{lemma}

\begin{proof} If not, there exist open sets $O_1$ and $O_2$ such that $K_\infty$ is the disjoint union
of $K_\infty\bigcap O_1$ and $K_\infty\bigcap O_2$. Since $K_\infty$
is compact and since both $K_\infty\bigcap O_1$ and $K_\infty\bigcap
O_2$ are closed in the relative topology of $K_\infty$, both
$K_\infty\bigcap O_1$ and $K_\infty\bigcap O_2$ are compact and,
therefore, there are disjoint open neighborhoods of the two disjoint
compact sets (a property of a Hausdorff space). Without loss of
generality, we can call the disjoint neighborhoods by $O_1$ and
$O_2$. Suppose $b_\infty\in O_2$. For $j$ large enough we have
$K_j\subset O_1\bigcup O_2$ and $b_j\in O_2$ and, therefore,
$K_j\subset O_2$ because $K_j$ is connected. This is a
contradiction.
\end{proof}

\noindent{\it Proof of Theorem \ref{thm1}.}  Assume ${\cal
S}_\gamma$ does not converge to ${\cal S}$ in Hausdorff metric. Then
there exist a sequence $\{p_j\}\subset{\cal S}$ and
$\{\gamma_j\}\to0$ such that $dist(p_j,{\cal
S}_{\gamma_j})>2\epsilon$ for some $\epsilon>0$.  Taking a limit
point $z\in{\cal S}$ of $\{p_j\}$ and choosing a subsequence if
necessary we can assume $dist(z,{\cal S}_{\gamma_j})>\epsilon$ for
all $j$'s. Such $z$ cannot be $\beta\in{\cal S} $ because
$\{\beta_{\gamma_j}\in{\cal S}_{\gamma_j}\}$ converges to $\beta$ as
$j$ goes to $\infty$.
 Since $\phi_{\gamma_j}: {\cal
S}_{\gamma_j}\setminus\{\beta_{\gamma_j},\overline{\beta_{{\gamma_j}}}\}\to
I\setminus\{0,-2\pi\mathrm{i}\}$ is invertible by Lemma
\ref{lemma20}, we can define
$$z_j:=\phi_{\gamma_j}^{-1}\circ\phi_{0}(z)\in{\cal S}_{\gamma_j}.$$ Let
$z_\infty$ be a limit point of $\{z_j\}$, then
$z_\infty\notin\{0,a\}$ because ${\cal S}_{\gamma_j}$ is uniformly
away from $0$ and $a$ for sufficiently small $\gamma_j.$ We also
have $z_\infty\neq\beta$ (and similarly,
$z_\infty\neq\overline\beta$) because, if not,
$|z_j-\beta_{\gamma_j}|$ would go to zero while
$|\phi_{\gamma_j}(z_j)-\phi_{\gamma_j}(\beta_{\gamma_j})|=|\phi_0(z)|>0$.

Since $({\rm clos}\,\{z_j\})\cap\{0,a\}=\emptyset$ Lemma \ref{lem1}
says that
$$  |\phi_0(z)-\phi_{0}(z_j)|=|\phi_{\gamma_j}(z_j)-\phi_{0}(z_j)|\stackrel{j\to\infty}{\longrightarrow} 0. $$
Since a subsequence of $\{\phi_{0}(z_j)\}$ converges to
$\phi_{0}(z_\infty)$ by the continuity of $\phi_0$, we have
\begin{equation}\label{phi0infty}
\phi_{0}(z)=\phi_{0}(z_\infty).
\end{equation}

Let ${\cal S}_\infty$ be the set of limit points of $\{{\cal
S}_{\gamma_j}\}$.  By Lemma \ref{lemma3} ${\cal S}_\infty$ is
connected to $\beta$.  Since ${\cal S}$ is the only component of
$\phi_0^{-1}(I)$ that is connected to $\beta$ we have ${\cal
S}_\infty\subset{\cal S}$. From \eqref{phi0infty} and $z_\infty\in
{\cal S}\setminus\{\beta,\overline\beta\}$, we get $z=z_\infty$ by
Lemma \ref{lemma20}.  This is a contradiction because $z_\infty$ is
a limit point of $\{{\cal S}_{\gamma_j}\}$ and, therefore,
$dist(z,z_\infty)\geq \epsilon$. This concludes the proof of ${\cal
S}_\gamma\to {\cal S}$.

\begin{figure}
\begin{center}
\includegraphics[width=0.45\textwidth]{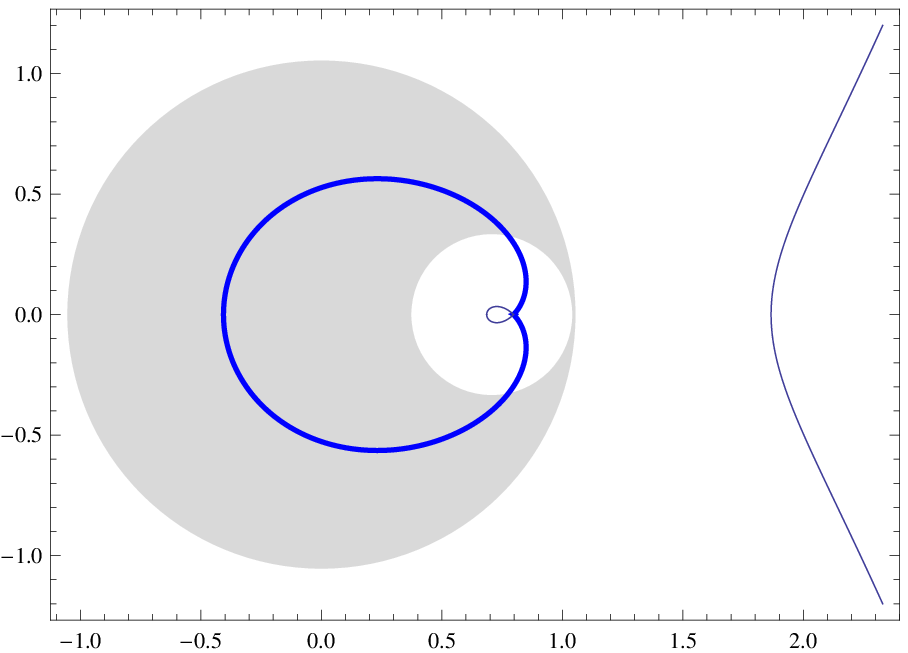}
\includegraphics[width=0.495\textwidth]{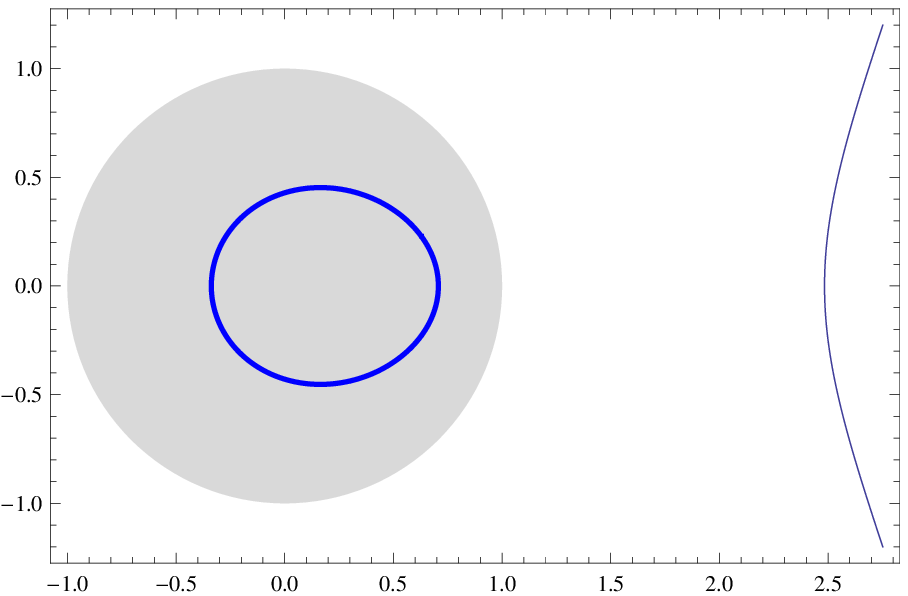}
\end{center}
\caption{Illustration of the convergence, ${\cal S}_\gamma\to{\cal
S}$ and $K_\gamma\to K$, when $a=1/\sqrt 2$.  For $\gamma=1/9$
(left), ${\cal S}_\gamma$ is drawn with thick line and the rest of
the set $\{z:{\rm Re}\,\phi_\gamma(z)=0\}$ is drawn with the thin
line; $K$ is the shaded region.  Same for $\gamma=0$
(right).}\label{figtp1}
\end{figure}

For $a<1$, the convergence of $K_\gamma $ to ${\rm clos}\,\DD$
follows from \eqref{ap1}.

For $a\geq1,$ we need to show that $\partial
K_\gamma=f_\gamma(\partial \mathbb{D})$ converges to
$\partial\mathbb{D}$. Recall that, as $\gamma$ goes to zero,
 $\alpha$ goes to $1/a$, $\kappa$ goes to zero and $\rho$ goes to 1.
It follows that $\lim_{\gamma\to 0} f_\gamma(v)=v$, which means
$K_\gamma\to{\rm clos}\,\DD$.

For all $a$, the convergence of $\mu_\gamma$ to $\mu$ follows from
the facts ${\cal S}_\gamma\to{\cal S}$ and
$\lim_{\gamma\to0}|y_\gamma(z)|=2\pi \rho(z)$ where $\rho$ is
defined in \eqref{mu}.\qed

\section{Matrix Riemann-Hilbert Problem}

The following fact is from \cite{Ba 2015}:
\begin{thm-others}\label{Thm2}
Let $\Gamma$ be a simple closed curve enclosing the line segment
$[0,a]\subset\CC$ and oriented counterclockwise.  Let the analytic
function $\omega_{n,N}$ on $\mathbb{C}\setminus{[0,a]}$ be defined
by
$$\omega_{n,N}(z) :=\left(\frac{z-a}{z}\right)^c\frac{e^{-Naz}}{z^{n}},$$
where  we choose the principal branch. Then the Riemann-Hilbert
problem,
\begin{equation}\nonumber
\left\{\begin{array}{lll}
Y(z) \text{ is holomorphic in }\mathbb{C}\setminus \Gamma, \\
\\
Y_+(z)=Y_-(z)\begin{bmatrix}
1&\omega_{n,N}(z) \\
0&1
\end{bmatrix} ,& z\in\Gamma,\\
\\\displaystyle
Y(z)=\left(I+\mathcal
{O}\left(\frac{1}{z}\right)\right)\begin{bmatrix}z^n&0\\0&z^{-n}\end{bmatrix},&
z\to\infty,
\end{array}\right.
\end{equation}
has the unique solution given by
$$ Y(z) = \begin{bmatrix} P_n(z) &\displaystyle \frac{1}{2\pi \mathrm{i}}\int_\Gamma\frac{P_n(w)\omega_{n,N}(w)}{w-z}dw \\
Q_{n-1}(z) &\displaystyle \frac{1}{2\pi
\mathrm{i}}\int_\Gamma\frac{Q_{n-1}(w)\omega_{n,N}(w)}{w-z}dw\end{bmatrix},
$$ where $Q_{n-1}(z)$ is the unique polynomial of degree $n-1$ such that
$$ \frac{1}{2\pi
\mathrm{i}}\int_\Gamma \frac{Q_{n-1}(w)\omega_{n,N}(w)}{w-z}dw
=\frac{1}{z^n}\left(1+\mathcal {O}\left(\frac{1}{z}\right)\right).$$
\end{thm-others}

\begin{lemma}\label{lemma1}
 For $a<1$, there exists a neighborhood $V$ of
$\overline{{\rm Int}\,{\cal S}}$ such that ${\rm Re}\,\phi(z)<0$ on
$V\setminus {\cal S}$ and the boundary of $V$ is a smooth Jordan
curve. For $a\geq1$, there exists a domain $V$ such that it contains
$\overline{{\rm Int}\,{\cal S}}\setminus\{\beta\}$ and its boundary,
$\partial V$, is a smooth Jordan curve that intersects $\beta$. Also
${\cal S}$ is smooth except at $\beta$, where it makes a corner with
the inner (i.e. towards ${\rm Int}\,{\cal S}$) angle $\pi/2$.
Lastly, ${\rm Re}\,\phi>0$ on $(\beta,a]$.
\end{lemma}
\begin{proof} From the definition \eqref{phi0} of $\phi$, ${\rm Re}\,\phi$
is harmonic function away from ${\cal S}$ and the origin. Since
${\rm Re}\,\phi(z)$ diverges to $-\infty$ as $z$ goes to $0$, ${\rm
Re}\,\phi(z)$ has to be negative everywhere in ${\rm
{Int}\,\mathcal{ S}}$ -- otherwise ${\rm Re}\,\phi(z)$ has a local
maximum in ${\rm Int}\,{\cal S}$, which is impossible. For $a<1$,
since the only critical point, $1/a$, of $\phi$ is away from ${\cal
S}$ and since ${\rm Re}\,\phi_A$ is harmonic in a neighborhood of
$\cal S$, ${\rm Re}\,\phi$ is negative in the vicinity of $\cal S$.
For $a\geq1$, since $\beta$ is the only critical point of $\phi_A$,
the claim in the lemma about the local shape of ${\cal S}$ near
$\beta$ and about $\partial V$ being intersecting $\beta$ follows by
the local analysis of the harmonic function ${\rm Re}\,\phi_A(z)$.
Specifically, ${\rm Re}\,\phi_A(z)$ is positive along the real axis
on $(0, \infty)\setminus\{\beta\}$, and is negative near $\beta$ in
the vertical direction (i.e. imaginary direction) from $\beta$.
\end{proof}
\begin{figure}
\begin{center}
\includegraphics[width=0.55\textwidth]{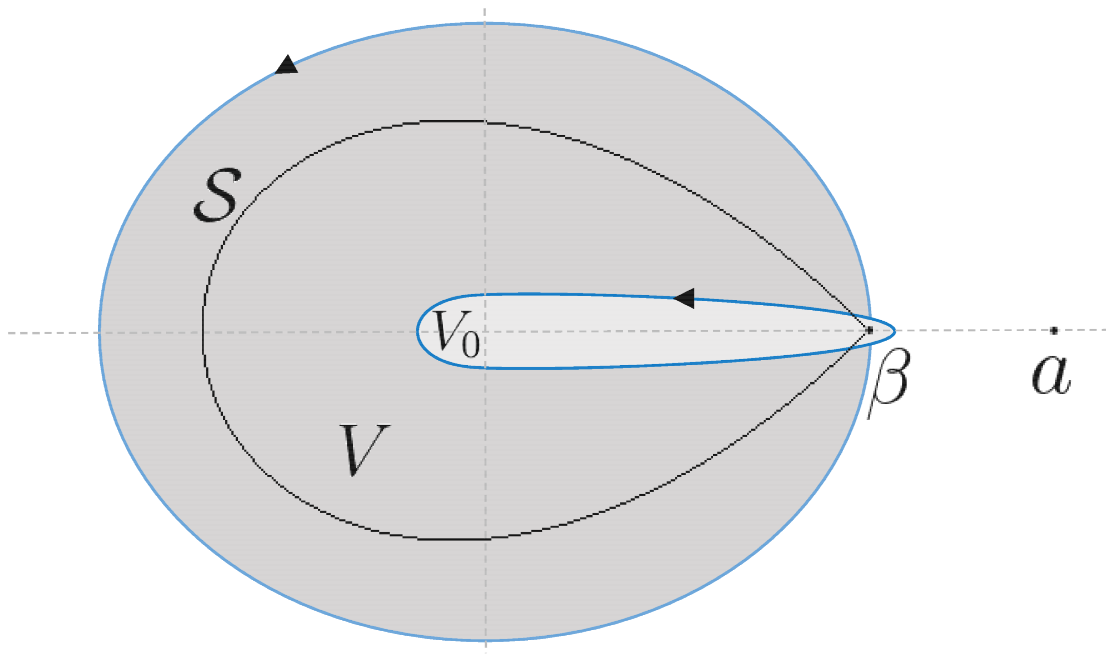}
\includegraphics[width=0.42\textwidth]{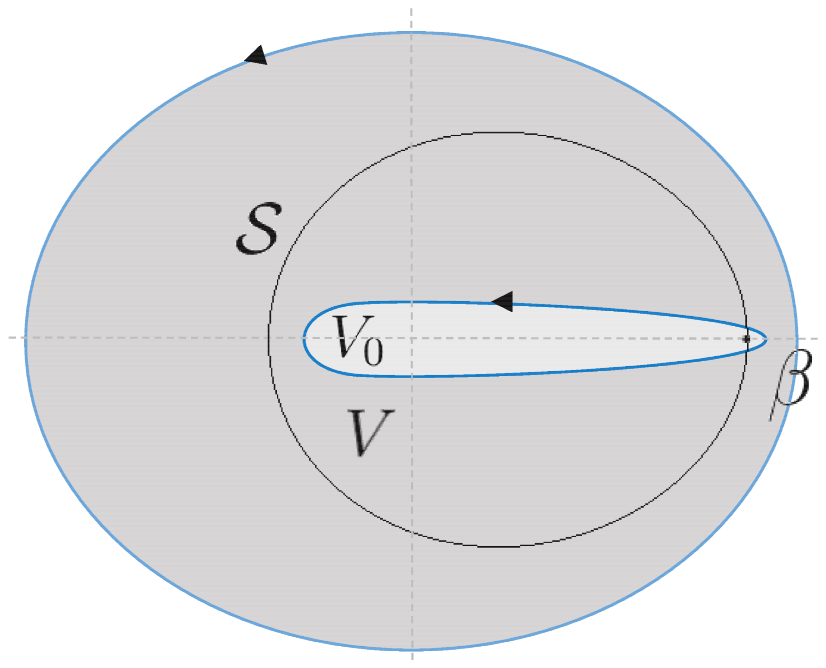}
\end{center}
\caption{ $V$ and $V_0$ for $a>1$ (left) and $a<1$ (right), ${\cal
S}$ is the black curve, $V$ is the interior of the contour enclosing
the shaded region, $V_0$ is the interior of the contour enclosing
the non--shaded region. These domains are used to define the domain
$U$ at \eqref{uv}. }\label{lens}
\end{figure}

\begin{figure}
\includegraphics[width=\textwidth]{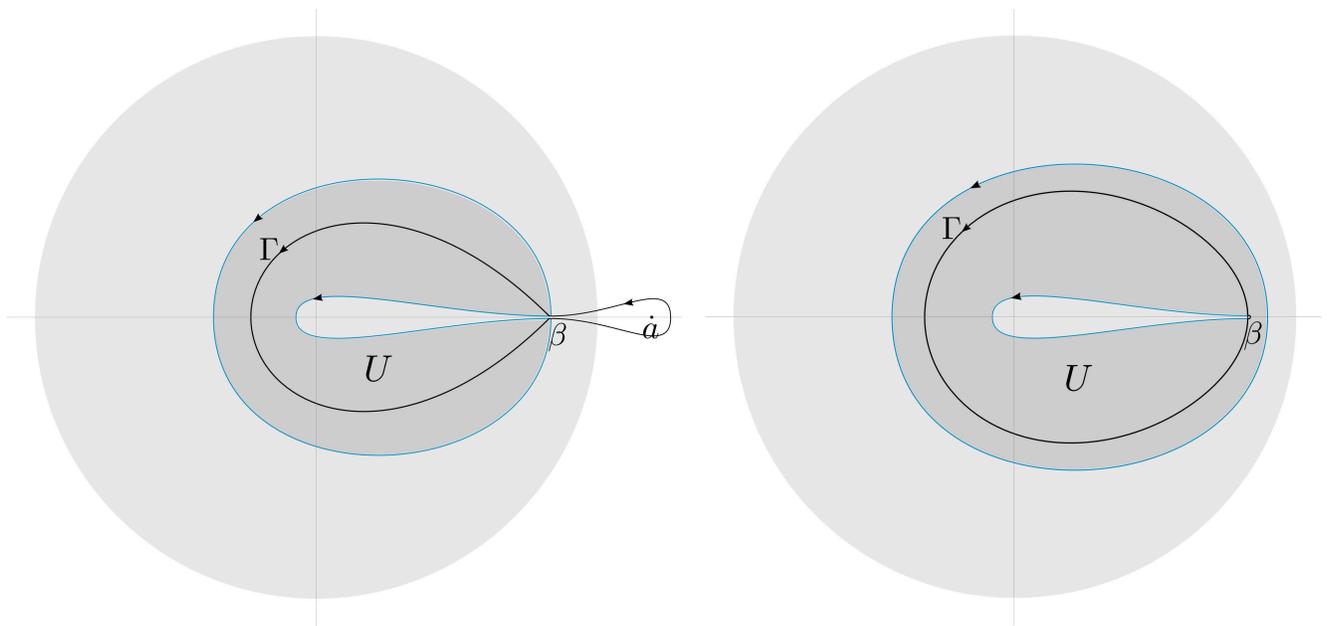}
\caption{ Contours for the Riemann-Hilbert problem of $\Phi$ when
$a>1$ (left) and $a<1$ (right). $\Gamma$ is the black curves and $U$
is the shaded region bounded by the blue curves.}\label{lens1}
\end{figure}

Using $V$ from the above lemma, we define the domain $U$ as below
\begin{equation}\label{uv}U=V\setminus \overline{V_0}.\end{equation} Here $V_0$ is a small open neighborhood of
$[0,\beta]$ such that its boundary, $\partial V_0$, is a smooth
Jordan curve that is arbitrarily close to $[0,\beta]$, see Figure
\ref{lens}. The region $U$ is the simply-connected (when $a\geq 1$)
or doubly-connected (when $a<1$) open neighborhood of ${\cal
S}\setminus \overline{V_0}$, disjoint from $[0,a]$ and with a
(piecewise) smooth boundary. We assign the counterclockwise
orientation on $\partial U\cap{\rm Ext}\,{\cal S}$ with respect to
the domain $U$ and the counterclockwise orientation on $\partial
U\cap{\rm Int}\,{\cal S}$ with respect to $V_0$.

From now on we let $\Gamma$ exactly match ${\cal S}$ inside $U$ and
away from a small neighborhood of $\beta$. When $a>1$, a part of the
contour $\Gamma$ goes outside $U$ around the line segment
$[\beta,a]$, see Figure \ref{lens1}. Near $\beta$ the reader should
not be concerned too much about the exact arrangement of  $\Gamma$
and $U$ as it will become clear when we define the local parametrix.

Below we define the complex logarithmic potential of $\mu$
\eqref{mu} by
$$g(z)=\int\log (z-w)\, d\mu(w), $$
where the specific branch of the $\log$ is chosen below. As a
function of $z$, this equals $\log z$ (modulo $2\pi\mathrm{i}$) when
$z\in{\rm Ext}\,{\cal S}$ by \eqref{balayage} and Theorem
\ref{thm1}, and has continuous real part, since the jump of $g$ on
${\cal S}$ is purely imaginary. These properties and \eqref{skele}
determine the explicit expression of this function as follows,
$$ g(z)=  \begin{cases} \log z ,\qquad & z \in \overline{{\rm Ext}\,{\cal S}},
\\ az+\log \beta-a\beta , &z\in{\rm Int}\,{\cal S}. \end{cases} $$
From the $g$-function above, we can write
\begin{equation}\nonumber
\phi(z) = az+\log z -2g(z) + \ell,\quad \ell = \log \beta-a\beta
\end{equation}
so that ${\rm Re}\,\phi(z)=0$ when $z\in{\cal S}$.

Following the standard nonlinear steepest descent method \cite{Deift
1999,DKMVZ 1999} applied to the matrix Riemann-Hilbert problem for
$Y$, we define $Z$ as the final object after the multiple transforms
of $Y$ given by
\begin{equation}\label{t1} Z(z)= e^{\frac{-N\ell}{2}\sigma_3}Y(z)\,e^{-Ng(z)\sigma_3}e^{\frac{N\ell}{2}\sigma_3}\begin{bmatrix}
1&0 \\\displaystyle \star\,
\Big(\frac{z}{z-a}\Big)^{c}e^{N\phi(z)}&1
\end{bmatrix}, \end{equation}
where
$$\star=\begin{cases} 1,\quad&\text{when $z\in U\cap {\rm Ext}\,{\Gamma}$,}\\-1,\quad&\text{when $z\in U\cap {\rm Int}\,{\Gamma}$,}\\0,\quad&\text{when $z\notin U$.}\end{cases} $$
Then $Z$ solves the following Riemann-Hilbert problem,
\begin{equation}\label{v1}
\left\{
\begin{array}{lll}
Z_+(z)=Z_-(z)\begin{bmatrix}
1&0 \\
\big(\frac{z}{z-a}\big)^{c}e^{N\phi(z)}&1
\end{bmatrix},& z\in\partial U , \vspace{0.2cm}\\
Z_+(z)=Z_-(z)\begin{bmatrix}
0&\big(\frac{z-a}{z}\big)^{c} \\
-\big(\frac{z}{z-a}\big)^{c}&0
\end{bmatrix} , & z\in\Gamma\cap U, \vspace{0.2cm}\\
Z_+(z)=Z_-(z)\begin{bmatrix}
1&\big(\frac{z-a}{z}\big)^{c}e^{-N\phi(z)} \\
0&1
\end{bmatrix} , & z\in \Gamma\setminus U.
\vspace{0.2cm}\\
Z(z)=I+\mathcal {O}(z^{-1}) , & z\rightarrow \infty.
\end{array}
\right.
\end{equation}

  We define
\begin{equation}\nonumber
\Phi(z)= \left\{
\begin{array}{lll}
\begin{bmatrix}
\displaystyle\Big(\frac{z}{z-\beta}\Big)^{c}&0 \\
0&\displaystyle\Big(\frac{z-\beta}{z}\Big)^{c}
\end{bmatrix}, & z\in {\rm Ext}\,{\Gamma},\vspace{0.2cm}\\
\begin{bmatrix}
0&\displaystyle\Big(\frac{z-a}{z-\beta}\Big)^{c} \\\displaystyle
-\Big(\frac{z-\beta}{z-a}\Big)^{c}&0
\end{bmatrix}, & z\in {\rm Int}\,{\Gamma},
\end{array}
\right.
\end{equation}
that satisfies the Riemann-Hilbert problem,
\begin{equation}\nonumber
\left\{
\begin{array}{lll}
\Phi_+(z)=\Phi_-(z)\begin{bmatrix}
0&\displaystyle\Big(\frac{z-a}{z}\Big)^{c} \\\displaystyle
-\Big(\frac{z}{z-a}\Big)^{c}&0
\end{bmatrix}, & z\in{\cal S},\vspace{0.2cm}\\\displaystyle
\Phi(z)=I+\mathcal {O}\left(\frac{1}{z}\right), & z\to\infty.
\end{array}
\right.
\end{equation}
Note that, when $a\leq 1$ and $z\in {\rm Int}\,{\cal S}$ we have
$\Phi(z) = \left[0~~1\atop-1~0\right]$. Also note that $\Phi$ is not
the only solution to the above Riemann-Hilbert problem -- for any
rational matrix function ${\cal R}(z)$ with a pole at $\beta$ such
that ${\cal R}(\infty)=I$, ${\cal R}(z)\Phi(z)$ is a solution.  We
will use this fact in the next section.

\section{$a>1$: when $c$ near $0$}\label{sec4}
 From the definition of $\phi_A$ at \eqref{phi0}, we obtain
\begin{equation*}
\phi_A(z)=\frac{a^2}{2}(z-\beta)^2\left(1+\mathcal
{O}(z-\beta)\right).
\end{equation*}
Let $D_{\beta}$ be a disk centered at $\beta$ such that there exists
a univalent map $\zeta: D_{\beta}\rightarrow\mathbb{C}$ as defined
in \eqref{zetamap}.
Under the mapping $\zeta$ the contour ${\cal S}$ maps into $[0,
e^{3\pi\mathrm{i}/4}t]\cup[0,
e^{-3\pi\mathrm{i}/4}t]_{t\in[0,\infty)}.$

In this section we intend to find $\mathcal
{P}:D_\beta\to\CC^{2\times 2}$ such that
\begin{equation}\label{samejump}Z^\infty(z)= \Phi(z)\left(\frac{z-a}{z}\right)^{\frac{c}{2}\sigma_3}\mathcal
 {P}(z)\left(\frac{z-a}{z}\right)^{-\frac{c}{2}\sigma_3},\quad z\in
 D_\beta
 \end{equation}
 satisfies
the jump condition of $Z$ at (\ref{v1}), i.e., we require ${\cal P}$
to satisfy, in $D_\beta$,
\begin{equation}\label{pj}
\left\{
\begin{array}{lll}
\mathcal {P}_{+}(z)=\mathcal {P}_{-}(z)\begin{bmatrix}
1&e^{-\zeta(z)^2/2} \\
0&1
\end{bmatrix}, & z\in \Gamma\setminus U,
\vspace{0.2cm}\\
\mathcal {P}_{+}(z)=\mathcal {P}_{-}(z)\begin{bmatrix}
1&0 \\
e^{\zeta(z)^2/2}&1
\end{bmatrix}, & z\in\partial U\cap{\rm Ext}\,\Gamma,
\vspace{0.2cm}\\
\mathcal {P}_{+}(z)=\mathcal {P}_{-}(z)\begin{bmatrix}
1&0 \\
e^{-\zeta(z)^2/2}&1
\end{bmatrix}, & z\in\partial U\cap{\rm Int}\,\Gamma,\vspace{0.2cm}\\
\mathcal {P}_{+}(z)=\begin{bmatrix}
0&-1 \\
1&0
\end{bmatrix}\mathcal {P}_{-}(z)\begin{bmatrix}
0&1 \\
-1&0
\end{bmatrix}, & z\in\Gamma\cap U,\vspace{0.2cm}\\
\mathcal {P}_{+}(z)=e^{-{c\pi\mathrm{i}}\sigma_3}\mathcal
{P}_{-}(z)e^{{c\pi\mathrm{i}}\sigma_3}, & z\in\RR ,\vspace{0.2cm}
\end{array}
\right.
\end{equation}
and the boundary condition, ${\cal P}(z)\sim I$ on $\partial
D_\beta$.   The fourth equation of \eqref{pj} comes from $\Phi$ in
\eqref{samejump} and the last equation comes from the (conjugating)
factors $\big((z-a)/z\big)^{\pm (c/2)\sigma_3}$ in \eqref{samejump}.
The jump contours, $\Gamma\setminus U$ and $\partial U\cap{\rm
Int}\,\Gamma$, can be pushed arbitrarily close to the real axis, so
that the jump contours of ${\cal P}$ consists of $\RR$, $i\RR$ and
$\{t \,e^{\pm \mathrm{i}3\pi/4}\}_{0<t<\infty}$. See Figure
\ref{pic5} for the illustration of the jump contours in $D_\beta$.

\begin{figure}
\includegraphics[width=\textwidth]{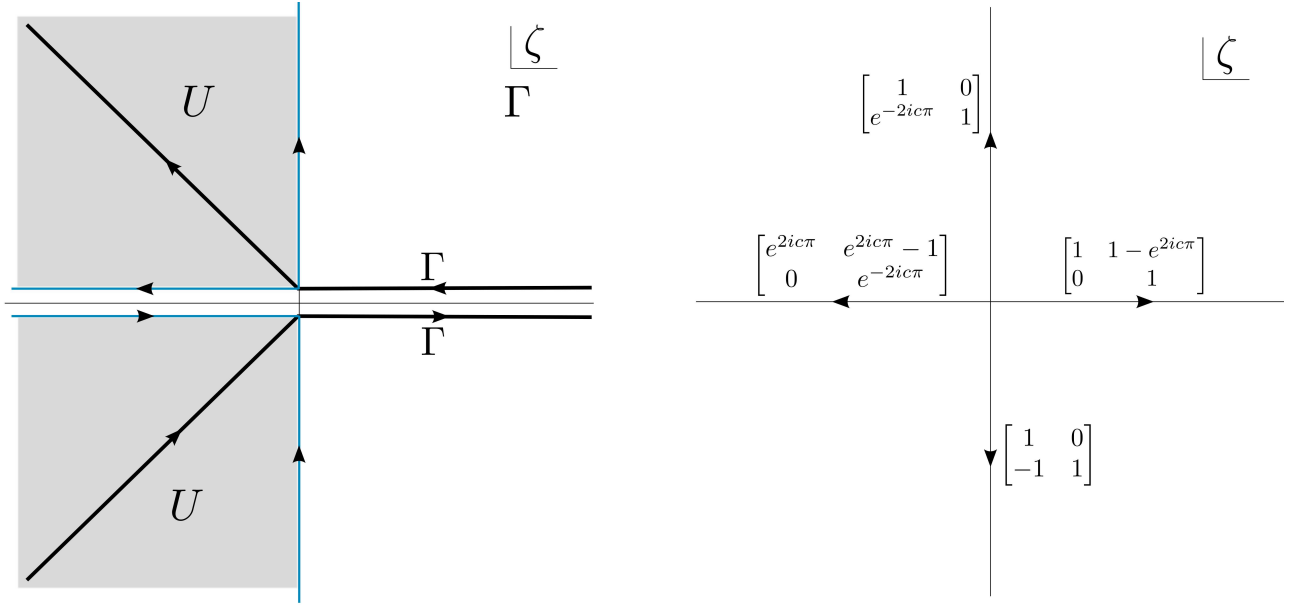}
\caption{Jump contours of ${\cal P}$ \eqref{pj} in $D_\beta$ (left)
and the jump matrices of $W$ (right)} \label{pic5}
\end{figure}

We want to transform ${\cal P}$ into a new matrix function, $W$,
that has only {\it constant jump matrices from the right}.   Such
transform may be given by
\begin{equation}\label{W} W(z):= \zeta(z)^{-{c}\sigma_3}S \cdot \mathcal
{P}(z)\cdot {T(\zeta(z))^{-1}\, S^{-1}},
\end{equation} using a diagonal matrix, $T$, and a piecewise
constant matrix, $S$, defined below, \begin{equation}\label{TS1}
T(\zeta) =\left\{\begin{array}{lll}\displaystyle \exp\left(
\frac{\zeta^2}{4}\sigma_3\right), & |\arg
\zeta|<3\pi/4,\vspace{0.4cm}\\\displaystyle
\exp\left(-\frac{\zeta^2}{4} \sigma_3\right), & \text{otherwise},
\end{array}\right.
\end{equation}
and
\begin{equation}\label{TS2} S=\left\{
\begin{array}{lll}
I,  & {\rm Im}\,{\zeta}<0\cap|\arg
\zeta|<3\pi/4,\vspace{0.2cm}\\
e^{c\pi\mathrm{i}\sigma_3}, & {\rm Im}\,{\zeta}>0\cap|\arg
\zeta|<3\pi/4,\vspace{0.2cm}\\
\begin{bmatrix}
0&1 \\
-1&0
\end{bmatrix}, & {\rm Im}\,{\zeta}<0\cap|\arg
\zeta|\geq3\pi/4,\vspace{0.2cm}\\
e^{c\pi\mathrm{i}\sigma_3}\begin{bmatrix}
0&1 \\
-1&0
\end{bmatrix}, & {\rm Im}\,{\zeta}>0\cap|\arg
\zeta|\geq3\pi/4.
\end{array}
\right.
\end{equation}
Here we choose $S$ such that $S^{-1}\zeta(z)^{c\sigma_3}$ satisfies
all the {\it left} jumps of ${\cal P}$, i.e.,
\begin{equation*}\begin{split}  \left(S^{-1}\zeta^{c\sigma_3}\right)_+ &=\begin{bmatrix}
0&-1 \\
1&0
\end{bmatrix}
\left( S^{-1}\zeta^{c\sigma_3}\right)_- ,  \,\quad z\in\Gamma\cap U,
 \\
\left(S^{-1}\zeta^{c\sigma_3}\right)_+ &= e^{-c\pi
\mathrm{i}\sigma_3} \left(S^{-1}\zeta^{c\sigma_3}\right)_- , \qquad
z\in\RR ,
\end{split}
\end{equation*}
such that $W$ has the jump matrices only from the {\it right}.
Furthermore, the jump matrices of $W$ are constant matrices because
of the right multipliction of $T^{-1}$ in \eqref{W}. The jump on
$\{t\,e^{\pm \mathrm{i}3\pi/4}\}_{0<t<\infty}$ disappears by the
right multiplication by $S^{-1}$.  We summarize the jump matrices of
$W$ below,
\begin{equation}\label{pj4}
W_{+}(z)=W_{-}(z)\left\{
\begin{array}{lll}
\begin{bmatrix}
1&1-e^{2\mathrm{i}c\pi} \\
0&1
\end{bmatrix},  & \zeta(z)\in\mathbb{R}^{+},\vspace{0.2cm}\\
\begin{bmatrix}
1&0 \\
e^{-2\mathrm{i}c\pi}&1
\end{bmatrix}, & \zeta(z)\in\mathrm{i}\mathbb{R}^{+},\vspace{0.2cm}\\
\begin{bmatrix}
e^{2\mathrm{i}c\pi}&e^{2\mathrm{i}c\pi}-1 \\
0&e^{-2\mathrm{i}c\pi}
\end{bmatrix}, & \zeta(z)\in\mathbb{R}^{-},\vspace{0.2cm}\\
\begin{bmatrix}
1&0 \\
-1&1
\end{bmatrix}, & \zeta(z)\in\mathrm{i}\mathbb{R}^{-}.
\end{array}
\right.
\end{equation}
The following fact can be checked by direct calculation.

\begin{lemma}\label{LemN} For $z\in D_\beta$ we have
\begin{equation*}
\Phi(z)\left(\frac{z-a}{z}\right)^{\frac{c}{2}\sigma_3}S^{-1}\zeta^{{c}\sigma_3}=
\left(N^{c/2} \eta(z)\right)^{\sigma_3},
\end{equation*}
where $\eta:D_\beta\to\CC$ ,
\begin{equation}\nonumber
\eta(z): =  \frac{e^{-ic\pi/2}}{N^{c/2}}
\left(\frac{a-z}{z}\right)^{\frac{c}{2}}\left(\frac{z\,\zeta(z)}{z-\beta}\right)^{c}
\end{equation}
is a nonvanishing $N$-independent analytic function in $D_\beta$.
\end{lemma}

Using the parabolic cylinder function \eqref{para!} we define ${\cal
W}:\CC\setminus(\RR\cup\mathrm{i}\RR)\to\CC^{2\times 2}$ as
\begin{equation}\label{ww}
{\cal W}(\zeta)=\left\{
\begin{array}{lll}
\begin{bmatrix}
D_{-{c}}(\zeta)&\frac{\mathrm{i}\sqrt{2\pi}e^{\frac{c\pi \mathrm{i}}{2}}}{\Gamma({c})}D_{-1+c}(\mathrm{i}\zeta) \\
-\frac{\Gamma({c}+1)}{\sqrt{2\pi}e^{{c\pi
\mathrm{i}}}}D_{-1-{c}}(\zeta)&e^{-\frac{c\pi
\mathrm{i}}{2}}D_{{c}}(\mathrm{i}\zeta)
\end{bmatrix}, & -\frac{\pi}{2}<\arg(\zeta)<0,\vspace{0.2cm}
\\
\begin{bmatrix}
D_{-{c}}(\zeta)&-\frac{\mathrm{i}\sqrt{2\pi}e^{\frac{3c\pi \mathrm{i}}{2}}}{\Gamma({c})}D_{-1+{c}}(-\mathrm{i}\zeta) \\
-\frac{\Gamma({c}+1)}{\sqrt{2\pi}e^{{c\pi
\mathrm{i}}}}D_{-1-{c}}(\zeta)&e^{\frac{c\pi
\mathrm{i}}{2}}D_{{c}}(-\mathrm{i}\zeta)
\end{bmatrix}, & 0<\arg(\zeta)<\frac{\pi}{2},\vspace{0.2cm}
\\
\begin{bmatrix}
e^{-{c\pi \mathrm{i}}}D_{-{c}}(-\zeta)&-\frac{\mathrm{i}\sqrt{2\pi}e^{\frac{3c\pi \mathrm{i}}{2}}}{\Gamma({c})}D_{-1+{c}}(-\mathrm{i}\zeta) \\
\frac{\Gamma(1+{c})}{\sqrt{2\pi}e^{{2c\pi
\mathrm{i}}}}D_{-1-{c}}(-\zeta)&e^{\frac{c\pi
\mathrm{i}}{2}}D_{{c}}(-\mathrm{i}\zeta)
\end{bmatrix},&\frac{\pi}{2}<\arg(\zeta)<\pi,\vspace{0.2cm}
\\
\begin{bmatrix}
e^{{c\pi \mathrm{i}}}D_{-{c}}(-\zeta)&\frac{\mathrm{i}\sqrt{2\pi}e^{\frac{c\pi \mathrm{i}}{2}}}{\Gamma({c})}D_{-1+{c}}(\mathrm{i}\zeta) \\
\frac{\Gamma(1+{c})}{\sqrt{2\pi}}D_{-1-{c}}(-\zeta)&e^{-\frac{c\pi
\mathrm{i}}{2}}D_{{c}}(\mathrm{i}\zeta)
\end{bmatrix},&\pi<\arg(\zeta)<\frac{3\pi}{2}.
\end{array}
\right.
\end{equation}

\begin{lemma}\label{lemma4}
There exists the asymptotic expansion of $D_{-c}(\zeta)$ given by
\begin{equation}\label{D}
D_{-c}(\zeta)=e^{-\frac{\zeta^2}{4}}\zeta^{-c}\left(\sum_{s=0}^{n-1}(-1)^s\frac{(c)_{2s}}{s!(2\zeta^2)^s}+\varepsilon_n(\zeta)\right),\quad|\rm
arg\,\zeta|< \frac{\pi}{2}.
\end{equation}
 There exists a constant $C>0$
independent of $c$ such that
$$|\varepsilon_n(\zeta)|\leq
C\left|\frac{(\frac{c}{2})_n(\frac{c+1}{2})_n}
{n!(\zeta^2)^{n}}\right|,\quad|\rm arg\,\zeta|< \frac{\pi}{2}.$$
Here, $(\cdot)_n$ is Pochhammer's Symbol defined by
$(x)_n=\Gamma(x+n)/\Gamma(x)$.
\end{lemma}
The proof of this lemma is in \ref{app2}. Though the lemma only
concerns $|\rm arg\,\zeta|<\pi/2$, this turns out to cover every
term that appears in ${\cal W}(\zeta)$ of \eqref{ww} and leads to
the following lemma.
\begin{lemma}\label{prop}
${\cal W}(\zeta(z))$ satisfies the jump of $W$ \eqref{pj4} and the
asymptotic behavior
\begin{equation}\label{F0} {\cal F}(\zeta):= {\cal W}(\zeta)\,\zeta^{{c}\sigma_3}e^{\frac{\zeta^2}{4}\sigma_3}=I+\frac{C_1}{\zeta}+\frac{C_2}{\zeta^{2}}
+{\cal O}\left(\frac{1}{\zeta^{3}}\right)
\end{equation}
as $|\zeta|$ goes to $\infty$, where
$$C_1=\begin{bmatrix}
0&\frac{\sqrt{2\pi}e^{\mathrm{i}\pi c}}{\Gamma(c)} \\
-\frac{\Gamma(c+1)}{\sqrt{2\pi}e^{\mathrm{i}\pi c}}&0
\end{bmatrix}\quad \text{and} \quad C_2=\begin{bmatrix}
-\frac{c(c+1)}{2}&0 \\
0&\frac{c(c-1)}{2}
\end{bmatrix}.$$
Moreover, as $c\to 0$ and $|\zeta|\to\infty$, we get
\begin{align}\label{f11}
&{\cal F}(\zeta)F_1(\zeta)^{-1}= I +
\begin{bmatrix} \displaystyle{\cal O}\left(c\,\zeta^{-2}\right)
&\displaystyle{\cal O}\left(c\,\zeta^{-3}\right)
\vspace{0.4cm}\\\displaystyle {\cal O}\left(\zeta^{-1}\right)
&\displaystyle {\cal O}\left(c\,\zeta^{-2}\right)
\end{bmatrix},
\\\label{f21} &
{\cal F}(\zeta)F_1(\zeta)^{-1}F_2(\zeta)^{-1}= I + {\cal
O}\left(\zeta^{-3}\right),
\end{align}
where
\begin{align}\label{f1}
&F_{1}(\zeta) = I+\frac{1}{\zeta}\begin{bmatrix}
0&\displaystyle\frac{\sqrt{2\pi}e^{\mathrm{i}\pi c}}{\Gamma(c)} \\
0&0
\end{bmatrix},\\\label{f2} &
F_{2}(\zeta) = I+\begin{bmatrix}
\displaystyle-\frac{c(c+1)}{2}\frac{1}{\zeta^2}
&\displaystyle\frac{1}{\zeta^3}\frac{\sqrt{2\pi}e^{\mathrm{i}\pi
c}c^2(c+1)^2}{4\Gamma(c+1)} \\\displaystyle
-\frac{\Gamma(c+1)}{\sqrt{2\pi}e^{\mathrm{i}\pi c}}\frac{1}{\zeta}
&\displaystyle\frac{c(c+1)}{2}\frac{1}{\zeta^2}
\end{bmatrix}.
\end{align}
The error bound in \eqref{f11} is uniform over $c\in[-1/2,1/2]$ as
$\zeta$ tends to infinity, and the error bound in \eqref{f21} is for
a fixed $c$.
\end{lemma}
\begin{proof} The proof of the jump is an exercise using the
following identities \cite{Wh 1996, Dai 2009}:
\begin{equation}\nonumber
\begin{array}{lll}
D_{-c}(\zeta)&=&\frac{\Gamma(1-c)}{\sqrt{2\pi}}\left[e^{\frac{-c\pi\mathrm{i}}{2}}D_{c-1}(\mathrm{i}\zeta)+
e^{\frac{c\pi\mathrm{i}}{2}}D_{c-1}(-\mathrm{i}\zeta)\right],\\
D_{-c}(\zeta)&=&e^{-c\pi\mathrm{i}}D_{-c}(-\zeta)+
\frac{\sqrt{2\pi}}{\Gamma(c)}e^{\frac{(1-c)\pi\mathrm{i}}{2}}D_{c-1}(-\mathrm{i}\zeta),\\
D_{-c}(\zeta)&=&e^{c\pi\mathrm{i}}D_{-c}(-\zeta) +
\frac{\sqrt{2\pi}}{\Gamma(c)}e^{\frac{(c-1)\pi\mathrm{i}}{2}}D_{c-1}(\mathrm{i}\zeta).
\end{array}
\end{equation}
The proof of the asymptotic behavior is based on Lemma \ref{lemma4}
about the asymptotic behavior of the parabolic cylinder function. By
Lemma \ref{lemma4}, letting $n=1$, we have
$$|\varepsilon_1(\zeta)|\leq C\left|\frac{c(c+1)}
{\zeta^2}\right|,\quad|\rm arg\,\zeta|< \frac{\pi}{2}.
$$
This leads to $D_{-{c}}(\zeta)=
e^{-\zeta^2/4}\zeta^{-c}\left(1+{\cal
O}\left(c(c+1)/\zeta^2\right)\right)$. Similarly, we can obtain the
asymptotic expression for $D_{-1+c}(\mathrm{i}\zeta),$
$D_{-1-{c}}(\zeta),$ and $D_{{c}}(\mathrm{i}\zeta)$ and we get
$${\cal F}(\zeta)= F_1(\zeta) + \begin{bmatrix}
\displaystyle{\cal
O}\left(\frac{c(c+1)}{\zeta^2}\right)&\displaystyle{\cal
O}\left(\frac{(c-1)(c-2)}{\zeta^3\Gamma(c)}\right) \\\displaystyle
{\cal O}\left(\frac{\Gamma(c+1)}{\zeta}\right)&\displaystyle{\cal
O}\left(\frac{c(c-1)}{\zeta^2}\right)
\end{bmatrix}.$$
This leads to \eqref{f11} using $\Gamma(c)=c^{-1}(1+{\cal O}(c))$.
Similarly, the equations \eqref{f21} and \eqref{F0} follow from
Lemma \ref{lemma4}.
\end{proof}

Let $H$ be a unimodular holomorphic matrix function on $D_\beta$. We
define $W$ by
\begin{equation}\label{WHW} W(z) = H(z) {\cal W}(\zeta(z)),\quad
z\in D_\beta.
\end{equation}

Combining \eqref{W}, \eqref{F0} and \eqref{WHW}, the expression in
\eqref{samejump} can be written as
\begin{align}\label{PHF}
&\nonumber
\textstyle\Phi(z)\left(\frac{z-a}{z}\right)^{\frac{c}{2}\sigma_3}\mathcal
 {P}(z)\left(\frac{z-a}{z}\right)^{-\frac{c}{2}\sigma_3}
 \\
&\nonumber=
\Phi(z)\left(\frac{z-a}{z}\right)^{\frac{c}{2}\sigma_3}S^{-1}\zeta^{{c}\sigma_3}H(z){\cal
W}(z)\,ST(\zeta(z))\left(\frac{z-a}{z}\right)^{-\frac{c}{2}\sigma_3}
\\
&\nonumber= \Phi(z)\left(\frac{z-a}{z}\right)^{\frac{c}{2}\sigma_3}
S^{-1}\zeta^{{c}\sigma_3}H(z)\,{\cal
F}(\zeta(z))\,\zeta(z)^{-{c}\sigma_3}e^{-\frac{\zeta(z)^2}{4}\sigma_3}ST(\zeta(z))
\left(\frac{z-a}{z}\right)^{-\frac{c}{2}\sigma_3}.
\end{align}
By \eqref{TS1}, \eqref{TS2} and Lemma \ref{LemN}, we obtain
$$\zeta^{-{c}\sigma_3}e^{\frac{-\zeta^2}{4}\sigma_3}ST(\zeta(z))\left(\frac{z-a}{z}\right)^{-\frac{c}{2}\sigma_3}
=\zeta^{-{c}\sigma_3}S\left(\frac{z-a}{z}\right)^{-\frac{c}{2}\sigma_3}=\left(N^{c/2}
\eta(z)\right)^{-\sigma_3} \Phi(z).$$ The above equations lead to
the following Lemma.
\begin{lemma}\label{lemma2}
When $z\in D_\beta$, we have
\begin{equation}\label{PHF}
\Phi(z)\left(\frac{z-a}{z}\right)^{\frac{c}{2}\sigma_3}\mathcal
 {P}(z)\left(\frac{z-a}{z}\right)^{-\frac{c}{2}\sigma_3}=\left(N^{c/2} \eta(z)\right)^{\sigma_3} H(z)\,{\cal F}(\zeta(z))
\left(N^{c/2} \eta(z)\right)^{-\sigma_3} \Phi(z).
\end{equation}
\end{lemma}

\begin{theorem}\label{thma1}
For $a>1$ and $-1/2\leq c\leq1/2$, we get
\begin{equation}\nonumber
P_N(z) =\left\{
\begin{array}{lll}
\displaystyle z^N\left(\frac{z}{z-\beta}\right)^{c} \left(1+\mathcal
{O}\left(\frac{1}{N^{c+1/2}}\right)\right),
& z\in {\rm Ext}\,{\cal S}\setminus (U\cup D_{\beta}),\vspace{0.4cm}\\
\begin{split}\displaystyle & z^N\bigg(\left(\frac{z}{z-\beta}\right)^{c}-
\frac{\sqrt{2\pi}(a^2-1)^c}{N^{1/2-c}a\Gamma(c)}\frac{
e^{N\phi_{A}(z)}}{(z-\beta)} \left(\frac{z-\beta}{z-a}\right)^{c} \vspace{0.2cm}\\
&\displaystyle\qquad+\mathcal
{O}\left(\frac{1}{N^{c+1/2}},\frac{e^{N\phi_{A}}}{N^{c+1/2}}\right)\bigg),
\end{split} &z\in U\setminus D_{\beta},\vspace{0.4cm}\\
 \displaystyle
z^N\left(\left(\frac{z\zeta}{z-\beta}\right)^{c}e^{\frac{\zeta^2(z)}{4}}D_{-c}(\zeta(z))+\mathcal
{O}\left(\frac{1}{N^{1/2}},\frac{1}{N^{2c+1/2}}\right)\right), &z\in
D_{\beta}.
\end{array}
\right.
\end{equation}
The error bounds are uniform in $c\in[-1/2,1/2]$. The big $\mathcal
{O}$ notation with multiple arguments is defined by $\mathcal {O}(A,
B)=\mathcal {O}(A)+\mathcal {O}(B).$
\end{theorem}

This theorem is similar to Theorem \ref{thumm} except that the range
of $c$ is restricted to $[-1/2,1/2]$ and the error bounds are
uniform in the range.

\begin{proof} Using $F_1$ in \eqref{f1} we can define a unimodular meromorphic matrix function with a simple pole at $\beta$ by

\begin{equation}{\label{rh}}
{\cal R}(z) = I + \frac{\sqrt{2 \pi } \left(a^2-1\right)^c
}{N^{1/2-c} a \,\Gamma (c)}\frac{1}{z-\beta}\begin{bmatrix} 0 & 1\\
0&0 \end{bmatrix}
\end{equation}
such that we can set
\begin{equation}\label{rh2}
 H(z) =\left(N^{c/2} \eta(z)\right)^{-\sigma_3} {\cal R}(z)
\left(N^{c/2} \eta(z)\right)^{\sigma_3} F_1(\zeta(z))^{-1},
\end{equation}
i.e., the above is unimodular and holomorphic at $\beta$.

Now we define the strong asymptotics of $Z$ that we will denote by
 \begin{equation}\label{z}
    Z^{\infty}(z):= \left\{
\begin{array}{ll}
{\cal R}(z)\Phi(z), \quad & z\notin D_{\beta},\vspace{0.2cm}
\\\displaystyle
\Phi(z)\left(\frac{z-a}{z}\right)^{\frac{c}{2}\sigma_3}\mathcal
{P}(z)\left(\frac{z-a}{z}\right)^{-\frac{c}{2}\sigma_3}, \quad &
z\in  D_{\beta},
\end{array}
\right.
\end{equation}
where the second line is given in Lemma \ref{lemma2}. We get
\begin{equation}
\begin{split}
Z_{+}^{\infty}(z)\left(Z_{-}^{\infty}(z)\right)^{-1}
&=\Phi(z)\left(\frac{z-a}{z}\right)^{\frac{c}{2}\sigma_3}\mathcal
{P}(z)\left(\frac{z-a}{z}\right)^{-\frac{c}{2}\sigma_3}\Phi^{-1}(z){\cal
R}^{-1}(z)\vspace{0.2cm}\\
&=\left(N^{c/2} \eta(z)\right)^{\sigma_3}H(z){\cal
F}(z)\left(N^{c/2} \eta(z)\right)^{-\sigma_3}{\cal R}^{-1}(z)\vspace{0.2cm}\\
&=\left(N^{c/2} \eta(z)\right)^{\sigma_3}H(z)\widehat{\cal
F}(\zeta)H^{-1}(z)\left(N^{c/2} \eta(z)\right)^{-\sigma_3},
\end{split}
\end{equation}
where, in the last line, we define
\begin{equation}\nonumber
\widehat{\cal F}(\zeta)={\cal F}(\zeta)F_1(\zeta)^{-1}.
\end{equation}
Defining the error matrix by
\begin{equation}\nonumber
{\cal E}(z):=Z^{\infty}(z)\,Z^{-1}(z).
\end{equation}
 we get
 \begin{equation}\label{EEZZ}
 \begin{split}
 {\cal E}_{+}(z){\cal E}^{-1}_{-}(z)&=Z^{\infty}(z)_{+}\left(Z_{-}^{\infty}(z)\right)^{-1}\vspace{0.2cm}\\
  &= \left(N^{c/2} \eta(z)\right)^{\sigma_3}H(z)\widehat{\cal
F}(\zeta)H^{-1}(z)\left(N^{c/2} \eta(z)\right)^{-\sigma_3} \vspace{0.2cm}\\
&=I+\begin{bmatrix}\displaystyle \mathcal
{O}\left(\frac{c}{N}\right)&\displaystyle\mathcal
{O}\left(\frac{c}{N^{3/2-c}}\right) \\\displaystyle \mathcal
{O}\left(\frac{1}{N^{1/2+c}}\right)&\displaystyle\mathcal
{O}\left(\frac{c}{N}\right)
\end{bmatrix}=I+{\cal O}\left(\frac{1}{N^{1/2+c}}\right),
 \quad z\in\partial D_\beta,
\end{split}
\end{equation}
where, in the last equality, we use the asymptotic behavior
\eqref{f11} for $\widehat{\cal F}(\zeta)={\cal
F}(\zeta)F_1(\zeta)^{-1}$, and the asymptotic behavior of $H$ given
below.
\begin{equation} \label{herror}
H= \begin{bmatrix} 1 & h(z)\\  0&1 \end{bmatrix} ,\qquad h(z)
=\frac{\sqrt{2 \pi } \left(a^2-1\right)^c }{\sqrt N\eta^2(z)a
\,\Gamma
(c)}\frac{1}{z-\beta}-\frac{1}{\zeta(z)}\frac{\sqrt{2\pi}e^{{\mathrm
i}\pi c}}{\Gamma(c)}={{\cal O}\left(\frac{c}{\sqrt N}\right).}
\end{equation}
One can check that the jump of ${\cal E}$ is exponentially small in
$N$ away from $\partial D_\beta$ using Lemma \ref{lemma1} and
\eqref{v1}. By the small norm theorem (e.g. Theorem 7.171 in
{\cite{Deift 1999} or \cite{DKMVZ 1999}}) we obtain ${\cal
E}(z)=I+\mathcal {O}\left(1/N^{c+1/2}\right)$
 and, therefore, $Z^{\infty}(z)Z^{-1}(z)=I+\mathcal
{O}\left(1/N^{c+1/2}\right)$.   Note that the error bound is uniform
over $c\in[-1/2,1/2]$.

Using \eqref{t1} we have (see \eqref{t1} for the definition of
$\star$)
\begin{equation}\nonumber
\begin{split}
Y(z)&=e^{\frac{N\ell}{2}\sigma_3}Z(z)
\begin{bmatrix}
1&0 \\\displaystyle -\star\,
\Big(\frac{z}{z-a}\Big)^{c}e^{N\phi(z)}&1
\end{bmatrix}
e^{\frac{-N\ell}{2}\sigma_3}e^{Ng(z)\sigma_3}
\\
&=e^{\frac{N\ell}{2}\sigma_3}\left(I+{\cal
O}\left(\frac{1}{N^{1/2+c}}\right)\right) Z^\infty(z)\begin{bmatrix}
1&0 \\\displaystyle -\star\,
\Big(\frac{z}{z-a}\Big)^{c}e^{N\phi(z)}&1
\end{bmatrix}e^{\frac{-N\ell}{2}\sigma_3}e^{Ng(z)\sigma_3}.
\end{split}
\end{equation}
Using \eqref{z}, we calculate the strong asymptotics for $z\in ({\rm
Ext}\,{\cal S}\cap U) \setminus D_{\beta}$ as an example.
 \begin{equation}\label{long1}
 \begin{split}
P_N(z)&=[Y(z)]_{11}=\left[\left(I+\mathcal
{O}\left(\frac{1}{N^{1/2+c}}\right)\right)Z^{\infty}(z)\begin{bmatrix}
1&0 \\ -\star\, \Big(\frac{z}{z-a}\Big)^{c}e^{N\phi(z)}&1
\end{bmatrix}\right]_{11}e^{Ng(z)}\vspace{0.2cm}\\
&=\left[\left(I+\mathcal
{O}\left(\frac{1}{N^{1/2+c}}\right)\right){\cal
R}(z)\Phi(z)\begin{bmatrix} 1&0 \\
-\Big(\frac{z}{z-a}\Big)^{c}e^{N\phi(z)}&1
\end{bmatrix}\right]_{11}e^{Ng(z)}\vspace{0.2cm}\\
&=\left[\left(I+\mathcal {O}\left(\frac{1}{N^{1/2+c}}\right)\right)
\begin{bmatrix} 1 & \frac{\sqrt{2 \pi } \left(a^2-1\right)^c }{N^{1/2-c} a \,\Gamma
(c)}\frac{1}{z-\beta}\\  0&1 \end{bmatrix}\begin{bmatrix}
\big(\frac{z}{z-\beta}\big)^{c}&0 \\
0&\big(\frac{z-\beta}{z}\big)^{c}
\end{bmatrix}\begin{bmatrix} 1&0
\\ -\big(\frac{z}{z-a}\big)^{c}e^{N\phi(z)}&1
\end{bmatrix}\right]_{11}z^N\vspace{0.2cm}\\
&=\bigg[\left(1+\mathcal {O}\left(\frac{1}{N^{1/2+c}}\right)\right)
\bigg(\left(\frac{z}{z-\beta}\right)^{c} -
\left(\frac{z-\beta}{z-a}\right)^{c}\frac{\sqrt{2\pi}(a^2-1)^c}{a\Gamma(c)N^{{1}/{2}-c}(z-\beta)}e^{N\phi(z)}
\bigg)
\\ & \qquad -\mathcal
{O}\left(\frac{1}{N^{1/2+c}}\right)\left(\frac{z-\beta}{z-a}\right)^{c}e^{N\phi(z)}\bigg]\,z^N\vspace{0.2cm}\\
&=z^N\bigg(\left(\frac{z}{z-\beta}\right)^{c} -
\left(\frac{z-\beta}{z-a}\right)^{c}\frac{\sqrt{2\pi}(a^2-1)^c}{a\Gamma(c)N^{{1}/{2}-c}(z-\beta)}e^{N\phi(z)}+\mathcal
{O}\left(\frac{1}{N^{1/2+c}}\right) \bigg).
\end{split}
\end{equation}
A similar calculation will give the following for $z\in ({\rm
Int}\,{\cal S}\cap U) \setminus D_{\beta}$:
$$ P_N(z)=e^{Ng(z)}\bigg(\left(\frac{z}{z-\beta}\right)^{c}e^{N\phi(z)} -
\left(\frac{z-\beta}{z-a}\right)^{c}\frac{\sqrt{2\pi}(a^2-1)^c}{a\Gamma(c)N^{{1}/{2}-c}(z-\beta)}+\mathcal
{O}\left(\frac{1}{N^{1/2+c}}\right) \bigg). $$

For $z\in ({\rm Ext}\,{\cal S}\setminus U)\cap D_{\beta}$ we
calculate the strong asymptotics using \eqref{z}, \eqref{F0} and
Lemma \ref{lemma2} to represent $\mathcal {P}$ in terms of ${\cal
W}$ \eqref{ww} and $H(z)$ \eqref{rh2}.
\begin{equation}\label{long2}
 \begin{split}
P_N(z)&=[Y(z)]_{11}=\left[\left(I+\mathcal
{O}\left(\frac{1}{N^{1/2+c}}\right)\right)Z^{\infty}
\begin{bmatrix}
1&0 \\ -\star\, \Big(\frac{z}{z-a}\Big)^{c}e^{N\phi(z)}&1
\end{bmatrix}
\right]_{11}e^{Ng(z)}\vspace{0.2cm}\\
&=\left[\left(I+\mathcal
{O}\left(\frac{1}{N^{1/2+c}}\right)\right)\Phi(z)\left(\frac{z-a}{z}\right)^{\frac{c}{2}\sigma_3}\mathcal
{P}(z)\left(\frac{z-a}{z}\right)^{-\frac{c}{2}\sigma_3}\right]_{11} z^N\vspace{0.2cm}\\
&=\left[\left(I+\mathcal {O}\left(\frac{1}{N^{1/2+c}}\right)\right)
\left(N^{c/2} \eta(z)\right)^{\sigma_3} H(z)\,{\cal F}(\zeta(z))
\left(N^{c/2} \eta(z)\right)^{-\sigma_3} \Phi(z)\right]_{11}z^N\vspace{0.2cm}\\
&=\bigg[\left(I+\mathcal {O}\left(\frac{1}{N^{1/2+c}}\right)\right)
\left(N^{c/2} \eta(z)\right)^{\sigma_3} \begin{bmatrix} 1 & h(z)\\
0&1 \end{bmatrix}\,\begin{bmatrix}
D_{-{c}}(\zeta)&\frac{\mathrm{i}\sqrt{2\pi}e^{\frac{c\pi \mathrm{i}}{2}}}{\Gamma({c})}D_{-1+c}(\mathrm{i}\zeta) \\
-\frac{\Gamma({c}+1)}{\sqrt{2\pi}e^{{c\pi
\mathrm{i}}}}D_{-1-{c}}(\zeta)&e^{-\frac{c\pi
\mathrm{i}}{2}}D_{{c}}(\mathrm{i}\zeta)
\end{bmatrix}
\vspace{0.2cm}\\
&\qquad
\cdot\,\zeta^{{c}\sigma_3}e^{\frac{\zeta^2}{4}\sigma_3}\left(N^{c/2}
\eta(z)\right)^{-\sigma_3}\begin{bmatrix}
\Big(\frac{z}{z-\beta}\Big)^{c}&0 \\
0&\Big(\frac{z-\beta}{z}\Big)^{c}
\end{bmatrix}\bigg]_{11}z^N\vspace{0.2cm}\\
&=\bigg[\left(\frac{z}{z-\beta}\right)^{c}\zeta(z)^{c}e^{\frac{\zeta(z)^2}{4}}\bigg(D_{-c}(\zeta)-h(z)\frac{\Gamma({c}+1)}{\sqrt{2\pi}e^{{c\pi
\mathrm{i}}}}D_{-1-{c}}(\zeta)\bigg)\left(1+{\cal
O}\left(\frac{1}{N^{1/2+c}}\right)\right)
\\&\qquad +\mathcal
{O}\left(\frac{1}{N^{1/2+2c}}\right)\bigg]z^N
\\&=\bigg[\left(\frac{z}{z-\beta}\right)^{c}\zeta(z)^{c}e^{\frac{\zeta(z)^2}{4}} D_{-c}(\zeta) +{\cal O}\left(\frac{1}{\sqrt{N}},\frac{1}{N^{1/2+2c}}\right)\bigg]z^N.
\end{split}
\end{equation}
We used \eqref{herror} at the last equality.  Note that the above
error bounds are uniform over $c\in[-1/2,1/2]$.

For the other regions we skip the calculations as they are similar.
\end{proof}

\section{$a>1$: Proof of Theorem \ref{thumm}}\label{rs}

The proof of Theorem \ref{thumm} is identical to the above proof of
Theorem \ref{thma1} except that we use different ${\cal R}$ and $H$
(hence different ${\cal P}$). The construction of ${\cal R}$ and $H$
will be more involved and will be useful for the next case of $a<1$
and, therefore, we will describe the construction in a more general
setting.

Here we describe how to construct $\cal R$ and $\cal P$ inductively
such that the jump, $Z^{\infty}_+(Z^{\infty}_-)^{-1}$, of
$Z^{\infty}$ is close to the identity up to ${\cal O}({N^{-L}})$ for
any given $L>0$.  The inductive method that we describe here
involves only algebraic manipulations -- such as the inverse of
relatively small matrices.

We introduce several notations that we will use in this section.

Let us recall that $\zeta$ is a univalent function in $D_\beta$ such
that $\zeta(\beta)=0$ and $N^{-\tau_a}\zeta(z)/(z-\beta)$
 is an $N$-independent and
non-vanishing holomorphic function where (we include the case,
$a<1$, for later)
\begin{equation}\nonumber
\tau_a =\bigg\{ \begin{array}{ll} 1/2 &\text{ for } a>1, \\ 1
&\text{ for } a<1. \end{array}
\end{equation}

The lemma below generalize the definition of $\widehat{\cal F}$ that
we used in the previous section.
\begin{lemma}\label{lem4}
Let ${\cal F}$ be a unimodular piecewise analytic matrix function
with the asymptotic expansion around $\infty$ given by
$${\cal F}=I+\frac{C_1}{\zeta}+\frac{C_2}{\zeta^2}+\cdots,$$
where $C_j$'s are constant $2\times2$ matrices. For any positive
integer $L$, there exists a positive number of $k$ and a
decomposition
\begin{equation}\label{ff}
 {\cal F}(\zeta)= \widehat{\cal F}(\zeta) F_k(\zeta)\cdots
F_1(\zeta),
\end{equation}
such that, for all $1\leq j\leq k,$ $ { F}_j$ is a rational function
with only singularity at the origin, $ { F}_j(\infty)=I$,
$F_j(\zeta)-I$ is nilpotent and
$$ \widehat{\cal F}(\zeta)=I+{\cal O}\left(\zeta^{-L}\right).$$
\end{lemma}
\begin{proof}  Assume $${\cal F}(\zeta)=I+\frac{C_0}{\zeta^m}+{\cal
O}\left(\frac{1}{\zeta^{m+1}}\right),\quad
C_0=\begin{bmatrix} c_{11}& c_{12}\\
c_{21}&c_{22}\end{bmatrix}.$$ Since $\det {\cal F}=1$, we have
$c_{11}+c_{22}=0$. One can write $C_0$ as the sum of three nilpotent
matrices, $C_0=N_1+N_2+N_3$, where
\begin{equation}\nonumber
N_1=\begin{bmatrix} c_{11}& -c^2_{11}\\
1&-c_{11}\end{bmatrix},\quad N_2=\begin{bmatrix} 0& c_{12}-c^2_{11}\\
0&0\end{bmatrix},\quad N_3=\begin{bmatrix} 0& 0\\
c_{21}-1&0\end{bmatrix}.
\end{equation}
We get
$${\cal F}(\zeta)\left(I+\frac{N_1}{\zeta^m}\right)^{-1}\left(I+\frac{N_2}{\zeta^m}\right)^{-1}\left(I+\frac{N_3}{\zeta^m}\right)^{-1}=I+{\cal O}\left(\frac{1}{\zeta^{m+1}}\right).$$
Using induction, this proves the lemma.
\end{proof}

Given $\{{F_k}\}_{k=1,2,\cdots}$, we will define $\{H_k\}$ and $\{
R_k\}$ inductively. Let $H_0=I$. Assume that $H_{k-1}$ is
holomorphic and non-vanishing at $\beta$, and $H_{k-1}(z)=I+{\cal
O}(1/N^{\tau_a})$. We define
\begin{equation}\label{fk}
\widetilde{{F}}_k(z):=\left(N^{\frac{c}{2}}\eta(z)\right)^{\sigma_3}
H_{k-1}(z)F_k(\zeta(z))H_{k-1}^{-1}(z)\left(N^{\frac{c}{2}}\eta(z)\right)^{-\sigma_3}.
\end{equation}
If ${F_k}$ satisfies the property described in Lemma \ref{lem4}, we
have the following truncated Laurent series expansion near $\beta$,
$$\widetilde{{F}}^{-1}_k(z)=N^{\frac{c}{2}\sigma_3}\left(I+\sum_{j=-\infty}^{m_k}\frac{A_j}{(z-\beta)^j}\right)N^{-\frac{c}{2}\sigma_3},$$
for some positive integer $m_k$ and some constant matrices
$\{A_j\}$. Given $\{A_j\}$, the lemma below constructs $\{R_k\}$
inductively.

\begin{lemma}\label{r}
Given $\widetilde{{F}}_k(z)$ as above, the unique rational matrix
function $R_k$ such that its only singularity is at $\beta$,
$R_k(\infty)=I$ and $R_k(z)\widetilde{{F}}^{-1}_k(z)$ is holomorphic
at $\beta$, is given by
$$R_k(z)=N^{\frac{c}{2}\sigma_3}\left(I+\sum_{j=1}^{m_k}\frac{B_j}{(z-\beta)^j}\right)N^{-\frac{c}{2}\sigma_3},$$
where, for a sufficiently large $N$, ${B_j}$'s are given by
$$[B_{m_k},B_{m_{k-1}},\cdots,B_1]=-[A_{m_k},A_{m_{k-1}},\cdots,A_1]\left(I+\widetilde{M}\right)^{-1}.$$
The $2m_k\times 2m_k$ matrix $\widetilde M$ is given in block form
by
$$\widetilde{M}=\begin{bmatrix}
A_{0}&A_{-1}&\cdots&A_{1-m_{k}} \\
A_1&A_{0}&\cdots&A_{2-m_k}\\
\vdots&\ddots &\ddots&\vdots\\
A_{m_{k-1}}& \cdots& A_1&A_0
\end{bmatrix}$$
and, for a sufficiently large $N$, $I+\widetilde{M}$ is invertible.
Moreover, $\det R_k\equiv1.$
\end{lemma}
\begin{proof} Let
$$M=\begin{bmatrix}
A_{m_k}&A_{m_{k-1}}&\cdots&A_1 \\
&A_{m_k}&\cdots&A_2\\
& &\ddots&\vdots\\
& & &A_{m_k}
\end{bmatrix},$$
In order to make $R_k(z)\widetilde{{F}}^{-1}_k(z)$ holomorphic at
$\beta$, we require all the pole terms of
$R_k(z)\widetilde{{F}}^{-1}_k(z)$ to vanish. We obtain
\begin{align}\label{a1}
&[B_{m_k},B_{m_{k-1}},\cdots,B_1]\cdot M=0,\\\label{a2}
&[B_{m_k},B_{m_{k-1}},\cdots,B_1](I+\widetilde{M})+[A_{m_k},A_{m_{k-1}},\cdots,A_1]=0,
\end{align}
where the first equation comes from the the poles of the orders
$2m_k, 2m_k-1, \cdots, m_k+1$, and the second equation comes from
the poles orders $m_k, m_k-1, \cdots, 1$.

We explain a useful bound on $A_j$'s. If $F_k(\zeta)=I+{\cal
O}(\zeta^{-m_k})$, then $F_k(\zeta(z))=I+{\cal O}(N^{-m_k \tau_a})$
on $\partial D_{\beta}$. Therefore, we have $A_j={\cal
O}(N^{-m_k\tau_a})$ and $\|\widetilde{M}\|={\cal
O}(N^{-m_k\tau_a})$.   Hence $I+\widetilde{M}$ is invertible for a
sufficiently large $N$ so that, from \eqref{a2}, we can obtain
\begin{equation}\nonumber
[B_{m_k},B_{m_{k-1}},\cdots,B_1]=-[A_{m_k},A_{m_{k-1}},\cdots,A_1]\left(I+\widetilde{M}\right)^{-1}.
\end{equation}
Let us show that \eqref{a1} is satisfied.  Since $F_k(\zeta)-I$ is
nilpotent, $\widetilde{{F}}^{-1}_k(z)-I$ is nilpotent and,
therefore,
$$\left(\sum_{j=-\infty}^{m_k}\frac{A_j}{(z-\beta)^j}\right)^2=0.$$
This implies $M^2=0$ and $M\widetilde{M}=-\widetilde{M}M$. Then,
\begin{equation}\nonumber
\begin{array}{lll} [B_{m_k},B_{m_{k-1}},\cdots,B_1]\cdot
M&=&-[A_{m_k},A_{m_{k-1}},\cdots,A_1]\left(I+\widetilde{M}\right)^{-1}\cdot
M \vspace{0.2cm}\\
&=&-[M]_{\text{1st row}}\left(I-\widetilde{M}+\widetilde{M}^2+\cdots\right)\cdot M\vspace{0.2cm}\\
&=&-\left[M\cdot\left(I-\widetilde{M}+\widetilde{M}^2+\cdots\right)\cdot
M\right]_{\text{1st row}}\vspace{0.2cm}\\
&=&-\left[MM-M\widetilde{M}M+M\widetilde{M}^2M+\cdots
\right]_{\text{1st row}}\vspace{0.2cm}\\
&=&-\left[MM+M^2\widetilde{M}+M^2\widetilde{M}^2+\cdots
\right]_{\text{1st row}}=0\end{array}
\end{equation}
The ``1st row'' means the 1st two rows or, equivalently, the 1st row
in the $2\times 2$ block matrix. Since
$R_k(z)\widetilde{{F}}^{-1}_k(z)$ is holomorphic at $\beta$ and
$\det \widetilde{{F}}^{-1}_k(z)\equiv1$, $\det R_k(z)$ is
holomorphic at $\beta.$ Since $\det R_k(\infty)=1$, we have $\det
R_k\equiv1$.

Now we show that $R_k$ is unique. Assume $\widetilde{R}_k$ also
satisfies all the conditions satisfied by $R_k$ in the lemma. Then,
${R}_k\widetilde{R}^{-1}_k$ is holomorphic away from $\beta$,
${R}_k(z)\widetilde{R}_k(z)^{-1}\to I$ as $z\to\infty$, and
${R}_k\widetilde{R}^{-1}_k=R_k\widetilde{{F}}^{-1}_k\big(\widetilde{R}_k\widetilde{{F}}^{-1}_k\big)^{-1}$
is holomorphic at $\beta$. Thus, ${R}_k=\widetilde{R}_k.$
\end{proof}
\begin{cor}\label{cor}
If $F_k(\zeta)=I+{\cal O}(\zeta^{-m})$, then
$N^{-\frac{c}{2}\sigma_3}R_k(z)N^{\frac{c}{2}\sigma_3}=I+{\cal
O}(N^{-\tau_am})$ when $z\in\partial D_\beta.$
\end{cor}
\begin{proof}
From $A_j={\cal O}(N^{-m\tau_a})$, $B_j={\cal O}(N^{-m\tau_a})$
follows.
\end{proof}

 Using $R_k(z)$ from the above lemma, we define $H_k(z)$ by
\begin{equation}\label{hk}
H_k(z)=\left(N^{\frac{c}{2}}\eta(z)\right)^{-\sigma_3}R_k(z)\widetilde{{F}}^{-1}_k(z)\left(N^{\frac{c}{2}}
\eta(z)\right)^{\sigma_3}H_{k-1}(z).\end{equation} Since $H_0=I$, by
induction, $H_k(z)$ is holomorphic at $\beta$ and unimodular. By
Corollary \ref{cor} we get
\begin{equation}\label{Hbound}
H_k(z)=I+{\cal O}(N^{-\tau_a}),\quad z\in\overline{D_\beta}.
\end{equation}

\begin{lemma}\label{thum3}
For $z\in\partial D_\beta,$
 we have $$Z_{+}^{\infty}(z)\left(Z_{-}^{\infty}(z)\right)^{-1}=\left(N^{c/2} \eta(z)\right)^{\sigma_3}H(z)\widehat{\cal
F}(\zeta)H^{-1}(z)\left(N^{c/2} \eta(z)\right)^{-\sigma_3}.$$
\end{lemma}
\begin{proof} We have
\begin{equation}\label{e}
\begin{split}
Z_{+}^{\infty}(z)\left(Z_{-}^{\infty}(z)\right)^{-1}
&=\Phi(z)\left(\frac{z-a}{z}\right)^{\frac{c}{2}\sigma_3}\mathcal
{P}(z)\left(\frac{z-a}{z}\right)^{-\frac{c}{2}\sigma_3}\Phi^{-1}(z){\cal
R}^{-1}(z)\vspace{0.2cm}\\
&=\left(N^{c/2} \eta(z)\right)^{\sigma_3}H(z){\cal
F}(z)\left(N^{c/2} \eta(z)\right)^{-\sigma_3}{\cal R}^{-1}(z)\vspace{0.2cm}\\
&=\left(N^{c/2} \eta(z)\right)^{\sigma_3}H(z)\widehat{\cal
F}(\zeta)H^{-1}(z)\left(N^{c/2} \eta(z)\right)^{-\sigma_3}.
\end{split}
\end{equation}
The first equality is from \eqref{z}, the second equality comes from
Lemma \ref{lemma2}, and the last equality follows from \eqref{ff}
and
\begin{equation}\label{hk1}
H=H_k=\left(N^{c/2} \eta\right)^{-\sigma_3} R_k\cdots R_1
\left(N^{c/2} \eta\right)^{\sigma_3} F_1^{-1}\cdots F_k^{-1},
\end{equation} which follows from the inductive definition of  $H_k$ at
\eqref{hk} with $H_0=I$. The theorem is proved using Lemma
\ref{lem4} and \eqref{Hbound}.
\end{proof}

\begin{proof}[Proof of Theorem \ref{thumm}]
  Contrary to
the proof of Theorem \ref{thma1}, all the error bounds will be for a
{\em fixed} $c$.

Here, we construct $\{R_j\}$ and $\{H_j\}$ inductively from the
initial data $R_1={\cal R}$ and $H_1=H$ where ${\cal R}$ and $H$
given in \eqref{rh} and \eqref{rh2}.

By \eqref{f2} with \eqref{fk} a calculation leads to,
\begin{equation}\nonumber
\begin{split}
\widetilde{{F}}_2(z)&=\left(N^{\frac{c}{2}}\eta(z)\right)^{\sigma_3}H_{1}(z)
F_2(\zeta(z))H_{1}^{-1}(z)\left(N^{\frac{c}{2}}\eta(z)\right)^{-\sigma_3} \vspace{0.2cm}\\
&=N^{\frac{c}{2}\sigma_3}\left(I+\begin{bmatrix} {\cal
O}\left(\frac{1}{N}\right)&{\cal
O}\left(\frac{1}{N^{3/2}}\right) \\
{\cal O}\left(\frac{1}{\sqrt N}\right)&{\cal
O}\left(\frac{1}{N}\right)
\end{bmatrix}\right)N^{-\frac{c}{2}\sigma_3}.
\qquad z\in\partial D_\beta.
\end{split}
\end{equation}
An estimate using $H_1=I+{\cal O}(N^{-1/2})$ in \eqref{Hbound} gives
the same result except the bound at ($12$)-entry above may be
relaxed to ${\cal O}(N^{-1})$. Then by Lemma \ref{r} we have
$$R_2(z)=N^{\frac{c}{2}\sigma_3}\left(I+\begin{bmatrix} {\cal O}\left(\frac{1}{N}\right)&{\cal
O}\left(\frac{1}{N}\right) \\
{\cal O}\left(\frac{1}{\sqrt N}\right)&{\cal O}\left(\frac{1}{\sqrt
N}\right)\end{bmatrix}\right)N^{-\frac{c}{2}\sigma_3}.$$ Using
$R_1={\cal R}$ with \eqref{rh} we get
$$R_2R_1=N^{\frac{c}{2}\sigma_3}\left(I+\begin{bmatrix} {\cal
O}\left(\frac{1}{N}\right)&\frac{\sqrt{2 \pi } \left(a^2-1\right)^c
}{\sqrt N a \,\Gamma (c)}\frac{1}{z-\beta}+{\cal
O}\left(\frac{1}{N}\right) \\
{\cal O}\left(\frac{1}{\sqrt N}\right)&{\cal O}\left(\frac{1}{\sqrt
N}\right)\end{bmatrix}\right)N^{-\frac{c}{2}\sigma_3}.$$ From
\eqref{f21} a further decompositions of ${\cal F}$ gives
${F}_k=I+{\cal O}\left(\zeta^{-3}\right)$ for $k\geq3$. Then, by
Corollary \ref{cor}, we get
$$R_k\cdots R_3=N^{\frac{c}{2}\sigma_3}(I+{\cal O}(N^{-3/2}))N^{-\frac{c}{2}\sigma_3},$$
and
$$R_k\cdots
R_1=N^{\frac{c}{2}\sigma_3}\left(I+\begin{bmatrix} {\cal
O}\left(\frac{1}{N}\right)&\frac{\sqrt{2 \pi } \left(a^2-1\right)^c
}{\sqrt N a \,\Gamma (c)}\frac{1}{z-\beta}+{\cal
O}\left(\frac{1}{N}\right) \\
{\cal O}\left(\frac{1}{\sqrt N}\right)&{\cal O}\left(\frac{1}{\sqrt
N}\right)\end{bmatrix}\right)N^{-\frac{c}{2}\sigma_3},\quad
z\in\partial D_\beta.$$ Using Lemma \ref{lem4}, we can have $
\widehat{\cal F}(\zeta)=I+{\cal O}\left(\zeta^{-L}\right)$ for an
arbitrary $L$. Using Lemma \ref{thum3} with
$${\cal R}=R_k\cdots R_1 \quad \text{and} \quad H=H_k= I+{\cal O}(N^{-1/2}),$$ we get
$Z_{+}^{\infty}\left(Z_{-}^{\infty}\right)^{-1}=I+{\cal O}(N^{-L})$
on $\partial D_\beta.$ From the similar argument as in the proof of
Theorem \ref{thma1}, we obtain
\begin{equation}\nonumber
Y(z)=e^{\frac{N\ell}{2}\sigma_3}\left(I+{\cal
O}\left(\frac{1}{N^{L}}\right)\right) Z^\infty(z)\begin{bmatrix} 1&0
\\\displaystyle -\star\, \Big(\frac{z}{z-a}\Big)^{c}e^{N\phi(z)}&1
\end{bmatrix}e^{\frac{-N\ell}{2}\sigma_3}e^{Ng(z)\sigma_3}
\end{equation}
uniformly over a compact set for an arbitrary positive integer $L$.
The proof is finished by calculations similar to \eqref{long1} and
\eqref{long2}.
\end{proof}

\section{$a<1$}\label{sec5}
In this section, we consider the case $a<1$ following closely the
analysis of previous two sections for the case $a>1$.

\begin{figure}
\begin{center}
\includegraphics[width=0.45\textwidth]{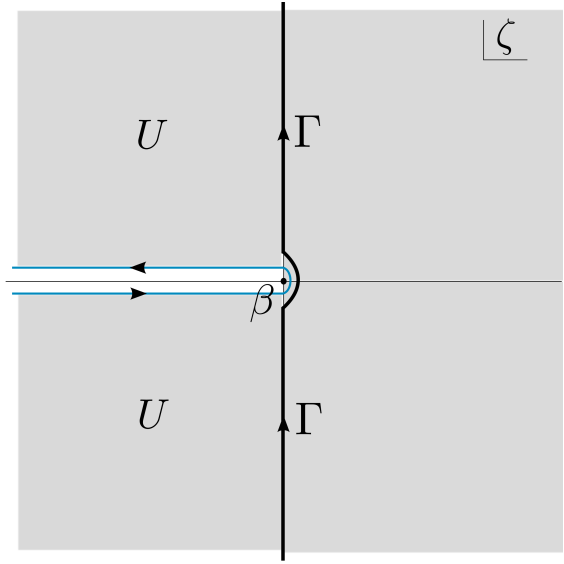}
\caption{Jump contours of ${\cal P}$ \eqref{post} in $D_\beta$
(left); the shaded region (everywhere except the negative real axis)
is $U$.}
\end{center}
\end{figure}

From \eqref{phi0}, we obtain
\begin{equation}\nonumber
\phi_{A}(z)=\frac{a^2-1}{a}(z-\beta)\left(1+\mathcal
{O}(z-\beta)\right).
\end{equation}
 We define  $\zeta: D_{\beta}\rightarrow\mathbb{C}$ by \eqref{zetamap}
where $D_{\beta}$ is a sufficiently small but fixed disc around
$z=\beta$ such that $\zeta$ is one-to-one. Under the mapping $\zeta$
the contour $\mathcal {S}$ maps to the imaginary axis.

Inside ${D}_{\beta}$ we want to find $\mathcal {P}$ such that
$$Z^\infty(z)=\Phi(z)\left(\frac{z-a}{z}\right)^{\frac{c}{2}\sigma_3}\mathcal
 {P}(z)\left(\frac{z-a}{z}\right)^{-\frac{c}{2}\sigma_3}$$ satisfies
the jump conditions of $Z$ in (\ref{v1}), i.e.,
\begin{equation}\label{post}
\left\{
\begin{array}{lll}
\mathcal {P}_{+}(z)=\mathcal {P}_{-}(z)\begin{bmatrix}
1&0 \\
e^{\zeta(z)}&1
\end{bmatrix},  &z\in\partial U\cap D_\beta,\vspace{0.2cm}\\
\mathcal {P}_{+}(z)=\begin{bmatrix}
0&-1 \\
1&0
\end{bmatrix}\mathcal {P}_{-}(z)\begin{bmatrix}
0&1 \\
-1&0
\end{bmatrix},& z\in\Gamma\cap D_\beta,\vspace{0.2cm}\\
\mathcal {P}_{+}(z)=e^{-{c\pi\mathrm{i}}\sigma_3}\mathcal
{P}_{-}(z)e^{{c\pi\mathrm{i}}\sigma_3}, &z\in(-\infty,a]\cap
D_\beta.
\end{array}
\right.
\end{equation}
Let us define $S$ by
$$S=S(\zeta)=\begin{cases} I,&|\arg \zeta|<\pi/2 ,\\\begin{bmatrix}
0&1 \\
-1&0
\end{bmatrix},\quad &\text{otherwise}. \end{cases}$$
Here we choose $S$ such that $S^{-1}\zeta(z)^{\frac{c}{2}\sigma_3}$
satisfies {\it the left jump of} $\mathcal {P}(z)$ from the second
and the third equations of (\ref{post}). Then the matrix function,
\begin{equation}\label{ww1}
W(z)=\zeta(z)^{-\frac{c}{2}\sigma_3}S\mathcal
{P}(z)S^{-1}\zeta(z)^{\frac{c}{2}\sigma_3},\end{equation} satisfies
\begin{equation}\nonumber
W_{+}(z)=W_{-}(z)\begin{bmatrix}
1&-\zeta(z)^{-c}e^{\zeta(z)} \\
0&1
\end{bmatrix}, \quad z\in\partial U\cap D_\beta.
\end{equation}
Let $H$ be a holomorphic matrix (that will be determined soon). A
solution to the above jump condition can be written by
$W(z)=H(z){\cal F}(\zeta(z))$ where \begin{equation}\label{fff}
 {\cal
F}(\zeta):=\begin{bmatrix}
1&\displaystyle\frac{-1}{2\mathrm{i}\pi}\int_{\cal L}\frac{e^{s}}{s^{c}(s-\zeta)}ds \\
0&1
\end{bmatrix}.
\end{equation}
Here the contour ${\cal L}$ is the image of $\partial U$ under
$\zeta$, and it begins at $-\infty$, circles the origin once in the
counterclockwise direction, and returns to $-\infty$.

\begin{lemma}\label{LemmaS}
For $z\in D_\beta$ we have
\begin{equation*}
\Phi(z)\left(\frac{z-a}{z}\right)^{\frac{c}{2}\sigma_3}S^{-1}\zeta(z)^{\frac{c}{2}\sigma_3}=
\left(N^{c/2} \eta(z)\right)^{\sigma_3}
\end{equation*}
where $\eta:D_\beta\to\CC$,
\begin{equation}\nonumber
\eta(z): =
\frac{1}{N^{c/2}}\left(\frac{z\,\zeta(z)}{z-\beta}\right)^{c/2},
\end{equation}
is a nonvanishing $N$-independent analytic function in $D_\beta$.
\end{lemma}

By Lemma \ref{LemmaS}, \eqref{ww1} and $W=H{\cal F}$ we get
\begin{align}\label{aa}
&\nonumber
\Phi(z)\left(\frac{z-a}{z}\right)^{\frac{c}{2}\sigma_3}\mathcal
 {P}(z)\left(\frac{z-a}{z}\right)^{-\frac{c}{2}\sigma_3}
 \\
&\nonumber=
\Phi(z)\left(\frac{z-a}{z}\right)^{\frac{c}{2}\sigma_3}S^{-1}\zeta^{{(c/2)}\sigma_3}{
W}(z)\,\zeta^{-{(c/2)}\sigma_3}S\left(\frac{z-a}{z}\right)^{-\frac{c}{2}\sigma_3}
\\
&\nonumber= \Phi(z)\left(\frac{z-a}{z}\right)^{\frac{c}{2}\sigma_3}
S^{-1}\zeta^{{(c/2)}\sigma_3}H(z)\,{\cal
F}(\zeta(z))\,\zeta^{-{(c/2)}\sigma_3}S
\left(\frac{z-a}{z}\right)^{-\frac{c}{2}\sigma_3}\\ &=\left(N^{c/2}
\eta(z)\right)^{\sigma_3} H(z)\,{\cal F}(\zeta(z)) \left(N^{c/2}
\eta(z)\right)^{-\sigma_3} \Phi(z).
\end{align}
This essentially proves the statement in Lemma \ref{lemma2} for
$a<1$.

\begin{lemma}\label{lemma10} As $|\zeta|$ goes to $\infty$,
${\cal F}$ in \eqref{fff} satisfies
 \begin{equation}\label{ffff}
 {\cal
F}(\zeta)F_1(\zeta)^{-1}= I + {\cal
O}\left(\frac{1}{|\zeta^2|}\right)
\end{equation}
uniformly over $c\in(-1,2)$ and
 \begin{equation}\label{ffff1}
{\cal F}(\zeta)F_1(\zeta)^{-1}\cdots F_k(\zeta)^{-1}= I + {\cal
O}\left(\frac{1}{|\zeta^{k+1}|}\right)
\end{equation}
where
\begin{equation}\label{Fkpost}
F_{k}(\zeta) = I+\frac{c_k}{\zeta^k}\begin{bmatrix}
0&1\\
0&0
\end{bmatrix},\quad c_k=\frac{1}{2\mathrm{i}\pi}\int_{\cal L}\frac{s^{k-1}e^{s}}{s^{c}}ds=\frac{\sin(c\pi)\,\Gamma(k-c)}{\pi(-1)^{k-1}}.
\end{equation}
\end{lemma}

 \begin{proof}
We only show the proof of \eqref{ffff} as the proof of \eqref{ffff1}
is similar.
 The only nonzero entry of $\left({\cal
F}F_1^{-1}-I\right)$ is the $(12)-$entry. For
$\arg|\zeta|<{\pi}/{2},$ we have
\begin{equation}\nonumber
\begin{split}
\left|\left({\cal
F}(\zeta)F_1(\zeta)^{-1}\right)_{12}\right|&=\frac{1}{2\pi}
\left|\int_{\cal L}\frac{e^{s}}{s^{c}(s-\zeta(z))}ds+\int_{\cal
L}\frac{e^{s}}{s^{c}\zeta(z)}ds\right|\\&\leq\frac{1}{2\pi}\int_{\cal
L}\left|\frac{e^{s}s}{s^{c}(\zeta(z)-s)\zeta(z)}\right||ds|
\leq\frac{1}{2\pi}
\int_{\cal L}\left|\frac{e^{s}s}{s^{c}\zeta^2}\right||ds|\vspace{0.2cm}\\
&= \frac{1}{2\pi |\zeta^2|} \int_{\cal
L}\left|\frac{e^{s}s}{s^{c}}\right||ds|.
\end{split}
\end{equation}
In the second inequality, we use $|\zeta-s|\geq|\zeta|$ for ${\rm
Re}\, \zeta>0$ and $s\in(-\infty,0].$ One can prove that the last
integral is finite by deforming the contour away from the origin so
that the integrant is bounded from above.

When $|\arg \zeta|\geq\pi/2$ a similar argument using the
deformation of integration contour leads to the proof of the lemma.
Note that the branch cut $(-\infty,0)$ of $s^c$ and the integration
contour ${\cal L}$ can be deformed, respectively, into
$\{te^{\mathrm{i}\theta_0}\}_{0< t<\infty}$ for
${\pi}/{2}\leq|\theta_0|\leq \pi$ and the corresponding contour
around the new branch cut. We skipped the further details.
\end{proof}

\begin{theorem}\label{thm6}
For $a<1$ we get
\begin{equation}\nonumber
P_N(z) =\left\{
\begin{array}{lll}
\displaystyle z^N\left(\frac{z}{z-a}\right)^{c}\left(1+\mathcal
{O}\left(\frac{1}{N^{2-c}}\right)\right), &z\in {\rm Ext}\,{\cal S}\setminus (U\cup D_{\beta}),\vspace{0.4cm}\\
\displaystyle
z^N\left(\left(\frac{z}{z-a}\right)^{c}-\frac{a(1-a^2)^{c-1}}{N^{1-c}\Gamma(c)}\frac{e^{N\phi_{A}(z)}}{(z-a)}
+\mathcal {O}\left(\frac{1 }{N^{2-c}},\frac{ e^{N\phi_{A}}
}{N^{2-c}}\right)\right), &z\in
U\setminus {D}_{\beta},\vspace{0.4cm}\\
\displaystyle
z^N\left(\left(\frac{z}{z-a}\right)^{c}-\left(\frac{z\zeta(z)}{z-a}\right)^{c}\frac{1}{e^{\zeta(z)}}\left({
\hat{f}}(\zeta(z))+{\cal O}\left(\frac{c}{N}\right)\right)+\mathcal
{O}\left(\frac{1}{N^{2-c}}\right)\right),  &z\in{D}_{\beta}.
\end{array}\right.
\end{equation}
where
$$\hat{f}(\zeta)=\frac{-1}{2\mathrm{i}\pi}\int_{\cal L}\frac{e^{s}}{s^{c}(s-\zeta)}ds.$$
Here the contour ${\cal L}$ is the image of $\partial U$ under
$\zeta$, and it begins at $-\infty$, circles the origin once in the
counterclockwise direction, and returns to $-\infty$. The error
bounds are uniform over $-1<c<2$.
\end{theorem}
\begin{proof} From $F_1$ in \eqref{Fkpost} one can obtain $R_1$ using Lemma \ref{r} and obtain $H_1$ by \eqref{hk1}:
\begin{equation}{\label{rh1}}
\begin{split}
 R_1(z) &= I + \frac{a(1-a^2)^{c-1}}{N^{1-c}\Gamma(c)}\frac{1}{z-a}
\begin{bmatrix} 0 & 1\\  0&0 \end{bmatrix}, \vspace{0.2cm}\\
 H_1(z) &=\left(N^{c/2} \eta(z)\right)^{-\sigma_3} R_1(z)
\left(N^{c/2} \eta(z)\right)^{\sigma_3} F_1(\zeta(z))^{-1} =
\begin{bmatrix} 1 & h(z)\\  0&1 \end{bmatrix},
\end{split}
\end{equation}
where (using $c_1=1/\Gamma(c)$ that appears in $F_1$)
\begin{equation}\label{hboundpost}
h(z) =
\left(\frac{z-a}{z\zeta(z)}\right)^{c}\left(\frac{a(1-a^2)^{c-1}}{N^{1-c}\Gamma(c)}\frac{1}{z-a}\right)-
\frac{1}{\zeta(z)\Gamma(c)}={\cal O}\left(\frac{c}{N}\right).
\end{equation}
Setting ${\cal R}=R_1$ and $H=H_1$, we can define $Z^\infty$ by
 \eqref{PHF} and \eqref{z}. Defining the error matrix by ${\cal
E}=Z^{\infty}Z^{-1}$, by the similar calculation as \eqref{EEZZ}
with $\widehat{\cal F}={\cal F}F_1^{-1}$ and \eqref{ffff}, we get
\begin{equation}\nonumber
{\cal E}_{+}(z){\cal E}^{-1}_{-}(z)=I+\mathcal
{O}\left(\frac{1}{N^{2-c}}\right),\quad z\in\partial D_\beta,
\end{equation}
uniformly over $c\in(-1,2)$.  By the same argument as in the proof
of Theorem \ref{thma1} we obtain
$$Z(z)=\left(I+\mathcal
{O}\left(\frac{1}{N^{2-c}}\right)\right)Z^{\infty}(z).$$ The proof
is finished by the calculations exactly similar to \eqref{long1} and
\eqref{long2}. To add a little more detail, inside $D_\beta$ we need
to use \eqref{hboundpost} to obtain the final result.  Below we
write the strong asymptotics {\em before} using  \eqref{hboundpost}
as an example.
$$\left(\left(\frac{z}{z-a}\right)^{c}-\left(\frac{z\zeta(z)}{z-a}\right)^{c}\left({
\hat{f}}(\zeta(z))+h(z)\right)e^{N\phi(z)}+\mathcal
{O}\left(\frac{1}{N^{2-c}}\right)\right)e^{Ng(z)},\qquad
z\in\text{Ext}\,\mathcal {S}\cap{D}_{\beta}.$$ We skip the
computation.
\end{proof}

\begin{proof}[Proof of Theorem \ref{thum33}]
The proof will be similar to the above proof and the proof of
Theorem \ref{thumm}.

By \eqref{Fkpost}, \eqref{rh1} and \eqref{fk} we get
\begin{equation}\label{f222}
\begin{split}
\widetilde{{F}}_2(z)&=\left(N^{\frac{c}{2}}\eta(z)\right)^{\sigma_3}H_{1}(z)
F_2(\zeta(z))H_{1}^{-1}(z)\left(N^{\frac{c}{2}}\eta(z)\right)^{-\sigma_3} \vspace{0.2cm}\\
&=N^{\frac{c}{2}\sigma_3}\left(I+\begin{bmatrix} 0&{\cal
O}\left(N^{-2}\right) \\
0&0
\end{bmatrix}\right)N^{-\frac{c}{2}\sigma_3},\qquad z\in\partial D_\beta.
\end{split}
\end{equation}
From Lemma \ref{r} and \eqref{f222} we have
$$R_2(z)=N^{\frac{c}{2}\sigma_3}\left(I+\begin{bmatrix} 0&{\cal
O}\left(N^{-2}\right) \\
0&0\end{bmatrix}\right)N^{-\frac{c}{2}\sigma_3}.$$ Combined with
${R_1}$ in \eqref{rh1}, we get
$$R_2R_1=N^{\frac{c}{2}\sigma_3}\left(I+\begin{bmatrix} 0&\displaystyle\frac{a(1-a^2)^{c-1}}{N\Gamma(c)}\frac{1}{z-a}+{\cal
O}\left(\frac{1}{N^2}\right) \\
0&0\end{bmatrix}\right)N^{-\frac{c}{2}\sigma_3}.$$ From
\eqref{Fkpost} in Lemma \ref{lemma10}, we have ${F}_k=I+{\cal
O}\left(\zeta^{-3}\right)$ for $k\geq3$. By Corollary \ref{cor}, we
obtain
$$R_k\cdots R_3=N^{\frac{c}{2}\sigma_3}(I+{\cal O}(N^{-3}))N^{-\frac{c}{2}\sigma_3}.$$
In fact, following the inductive construction of $R_k$ and $H_k$ in
 Section \ref{rs}, one can find that $R_k$'s are all upper diagonal matrix.
Therefore, we get
 $$R_k\cdots
R_1=N^{\frac{c}{2}\sigma_3}\left(I+\begin{bmatrix}
0&\displaystyle\frac{a(1-a^2)^{c-1}}{N\Gamma(c)}\frac{1}{z-a}+{\cal
O}\left(\frac{1}{N^2}\right) \\
0&0\end{bmatrix}\right)N^{-\frac{c}{2}\sigma_3},\quad z\in\partial
D_\beta.$$ Using Lemma \ref{lem4}, we can have $ \widehat{\cal
F}(\zeta)=I+{\cal O}\left(\zeta^{-L}\right)$ for an arbitrary $L$.
Using Lemma \ref{thum3} with
$${\cal R}=R_k\cdots R_1 \quad \text{and} \quad H=H_k= I+{\cal O}(N^{-1}),$$ we get
$Z_{+}^{\infty}\left(Z_{-}^{\infty}\right)^{-1}=I+{\cal O}(N^{-L})$
on $\partial D_\beta.$ From the similar argument as in the proof of
Theorem \ref{thma1}, we obtain
\begin{equation}\nonumber
Y(z)=e^{\frac{N\ell}{2}\sigma_3}\left(I+{\cal
O}\left(\frac{1}{N^{L}}\right)\right) Z^\infty(z)\begin{bmatrix} 1&0
\\\displaystyle -\star\, \Big(\frac{z}{z-a}\Big)^{c}e^{N\phi(z)}&1
\end{bmatrix}e^{\frac{-N\ell}{2}\sigma_3}e^{Ng(z)\sigma_3}.
\end{equation}
for an arbitrary positive integer $L$.  The proof is finished by
calculations similar to \eqref{long1} and \eqref{long2}.
\end{proof}

\section{Critical case: $a=1$}

\begin{figure}
 \begin{center}
 \includegraphics[width=0.5\textwidth]{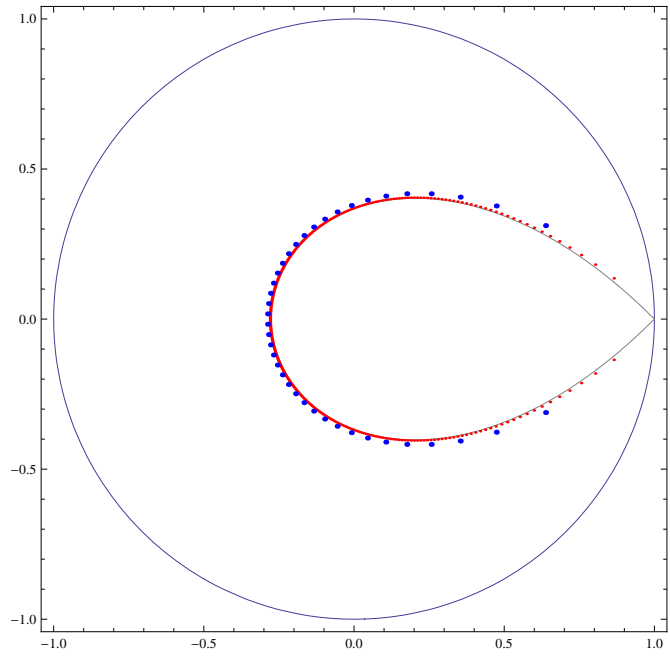}
 \caption{The zeros of orthogonal polynomials with degrees $40$ (blue) and $300$ (red),
$c=1$ and $a=1$. The solid line inside the disk is  ${\cal
S}$.}\label{pic2}
\end{center}
 \end{figure}

 \begin{figure}
\begin{center}
\includegraphics[width=0.5\textwidth]{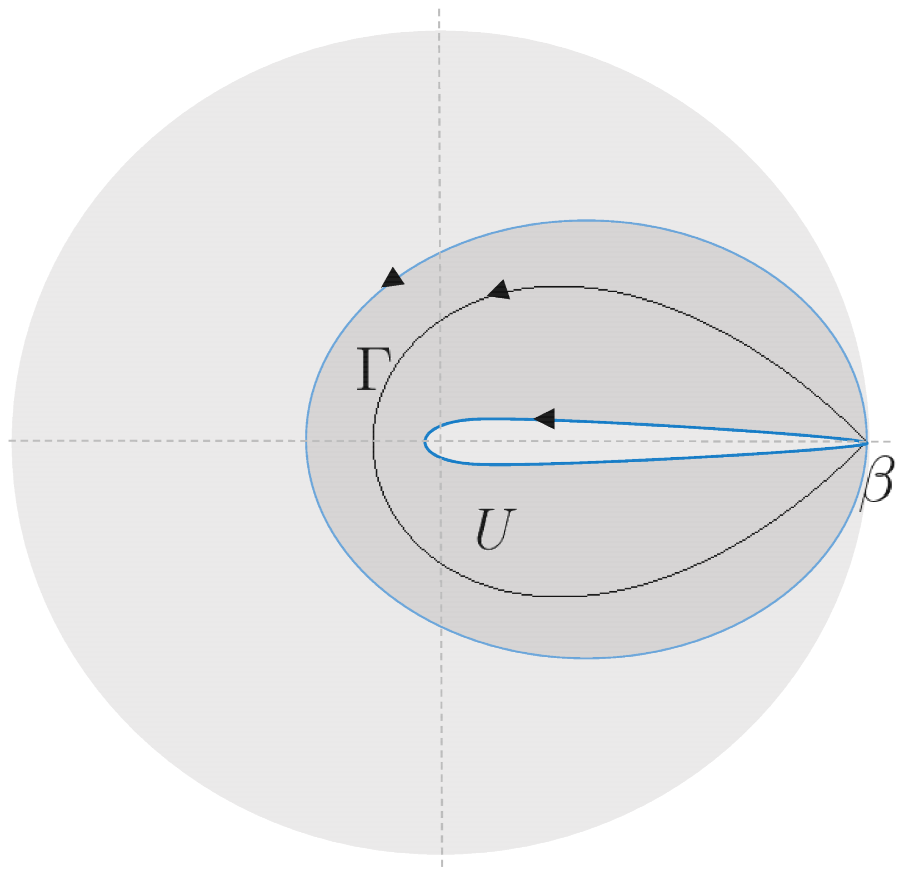}
\end{center}
\caption{
Contours for the Riemann-Hilbert problem of $\Phi$ when $a\approx1$.
$\Gamma$ is the black curves and $U$ is the shaded region bounded by
the blue curves. }\label{fig12}
\end{figure}

In this section we consider $a=1+{\cal O}(1/\sqrt N)$. Here we only
argue that the strong asymptotics {\em can} be obtained through the
parametrix of Painlev\'e IV equation (as suggested in \cite{Ba1
2015}) following the similar steps described previously.

There is a disk ${D}_{1}$ centered at $1$ such that there exists a
univalent map $\zeta: D_{1}\rightarrow\mathbb{C}$ that satisfies
\begin{equation}\nonumber
\left(\zeta(z)+x\right)^2=N\phi_{A}(z)- N\phi_{A}\left(1/a\right)
\end{equation}
where
$$x :=\sqrt{N\phi_{A}(a)-N\phi_{A}(1/a)}=\sqrt{2N}(a-1)(1+{\cal O}(a-1)).$$
Under the mapping $\zeta$, we have $\zeta(a)=0$ 
and the critical point of $\phi_{A}$ maps to $-x$; note that $\phi(1/a)$ is the {\em critical value} of $\phi_{A}$.  


Inside ${D}_1$ we require that
$\Phi(z)\left(\frac{z-a}{z}\right)^{\frac{c}{2}\sigma_3}\mathcal
{P}(z)\left(\frac{z-a}{z}\right)^{-\frac{c}{2}\sigma_3}$ satisfies
the jump conditions (\ref{v1}) of $Z$.  With the boundary condition
of ${\cal P}$ on $\partial D_1$ this leads to the following jumps of
${\cal P}$ {\em inside $D_1$}:
\begin{equation}\nonumber
\left\{
\begin{array}{lll}
\mathcal {P}_{+}(z)=\mathcal {P}_{-}(z)\begin{bmatrix}
1&0 \\
e^{-N\phi_{A}(z)}&1
\end{bmatrix}, & z\in\partial U\cap\text{Int}\,\Gamma,\vspace{0.2cm}\\
\mathcal {P}_{+}(z)=\mathcal {P}_{-}(z)\begin{bmatrix}
1&0 \\
e^{N\phi_{A}(z)}&1
\end{bmatrix}, & z\in\partial U\cap\text{Ext}\,\Gamma,\vspace{0.2cm}\\
\mathcal {P}_{+}(z)=\begin{bmatrix}
0&-1 \\
1&0
\end{bmatrix}\mathcal {P}_{-}(z)\begin{bmatrix}
0&1 \\
-1&0
\end{bmatrix}, & z\in\Gamma\cap U,\vspace{0.2cm}\\
\mathcal {P}_{+}(z)=e^{-{c\pi\mathrm{i}}\sigma_3}\mathcal {P}_{-}(z)e^{{c\pi\mathrm{i}}\sigma_3}, & z\in (0, a ],\vspace{0.2cm}\\
\mathcal {P}(z)=I+o\left(1\right), & z\in\partial {D}_{1}.
\end{array}
\right.
\end{equation}
Here $U$ and $\Gamma$ are given similarly by those for $a>1$ except
the segment $[\beta,a]$ becomes a point at $1$, see Figure
\ref{fig12}. We will show that such ${\cal P}$ can be written in
terms of the solution of the Painlev\'e IV equation.  To achieve
this, we want to transform ${\cal P}$ into a new matrix function,
$W$, with only {\it constant jump matrices from the right}.   Such
transform is given by
\begin{equation}\label{W1} W(z)= e^{-\frac{\ell_x}{2}\sigma_3}\zeta(z)^{\frac{c}{2}\sigma_3}S \cdot \mathcal
{P}(z)\cdot {T(z)^{-1}\, S^{-1}},\quad z\in D_1,
\end{equation} using a diagonal matrices $T$, a
piecewise constant matrix $S$ and a constant $\ell_x$, defined by
\begin{equation}\nonumber
\begin{split}
 T(z)&=\exp\left(\frac{N}{2}(-1)^\nu\phi_A(z)\sigma_3\right)=\exp\left[ \frac{(-1)^{\nu}}{2}
\left(\zeta(z)^2+2x\zeta(z)+\ell_x\right)\sigma_3\right],\quad
\\
\ell_x  &= x^2 + N\phi_{A}(1/a) ,\quad S=S(z)=
\begin{bmatrix}
0&1 \\
1&0
\end{bmatrix}
\cdot
\begin{bmatrix}
0&1 \\
-1&0
\end{bmatrix}^{\nu},
\end{split}
\end{equation}
where
$$ \nu=\begin{cases} 0, &z\in{\rm Ext}\,\Gamma,\\1,\quad &z\in{\rm Int}\,\Gamma. \end{cases}  $$
Here we chose $S$ such that $S^{-1}\zeta(z)^{-\frac{c}{2}\sigma_3}$
satisfies all the {\it left} jumps of ${\cal P}$, i.e.,
\begin{equation*}\begin{split}  \left(S^{-1}\zeta(z)^{-\frac{c}{2}\sigma_3}\right)_+ &=\begin{bmatrix}
0&1 \\
-1&0
\end{bmatrix}
\left(S^{-1}\zeta(z)^{-\frac{c}{2}\sigma_3}\right)_- ,  \quad
z\in\Gamma\cap U,
 \\
\left(S^{-1}\zeta^{-\frac{c}{2}\sigma_3}\right)_+ &= e^{-c\pi
i\sigma_3} \left(S^{-1}\zeta^{-\frac{c}{2}\sigma_3}\right)_- ,
\quad\qquad z\in[-\infty, 0] .
\end{split}
\end{equation*}
Consequently, $W$ has the jump matrices only from the {\it right}.
Furthermore, the jump matrices of $W$ are constant matrices because
of the right multipliction of $T$ in \eqref{W1}, and the jump on
$\Gamma$ disappears by the right multiplication by $S^{-1}$.  We
obtain the jump condition of $W$ by
\begin{equation}\nonumber
W_{+}(z)=W_{-}(z)\left\{
\begin{array}{lll}
\begin{bmatrix}
1&0 \\
s_1&1
\end{bmatrix},  & \zeta(z)\in\mathbb{R}^{+},\vspace{0.2cm}\\
\begin{bmatrix}
1&s_2 \\
0&1
\end{bmatrix}, & \zeta(z)\in\mathrm{i}\mathbb{R}^{+},\vspace{0.2cm}\\
\begin{bmatrix}
1&0 \\
s_3&1
\end{bmatrix}, & \zeta(z)\in\mathbb{R}^{-},\vspace{0.2cm}\\
\begin{bmatrix}
1&s_4 \\
0&1
\end{bmatrix}, & \zeta(z)\in\mathrm{i}\mathbb{R}^{-},
\end{array}
\right.
\end{equation}
where $s_1=0$, $s_2=1$, $s_3=e^{2\mathrm{i}c\pi}-1$ and
$s_4=-e^{-2\mathrm{i}c\pi}$.  The boundary condition at $\partial
D_1$ gives
$$W(z)=\zeta(z)^{\frac{c}{2}\sigma_3}\left(I+o\left(1\right)\right)e^{\big(\frac{\zeta(z)^2}{2}+x\zeta(z)\big)
\sigma_3}, \quad z\in\partial D_{\beta}.$$ Here we used that
$\ell_x={\cal O}(1)$ for $a=1+{\cal O}(1/\sqrt N)$. According to
page 34 of \cite{Dai 2009} (or \cite{FO 2006}) the Riemann-Hilbert
problem for the Painlev\'e IV parametrix -- $\Psi$, following the
notation in \cite{Dai 2009} -- exactly satisfies the jump condition
above and the boundary condition:
$$\Psi(\zeta,x)=\left(I+\frac{\Psi_{-1}(x)}{\zeta}+\frac{\Psi_{-2}(x)}{\zeta^2}+\mathcal
{O}\left(\frac{1}{\zeta^3}\right)\right)e^{\big(\frac{\zeta^2}{2}+x\zeta\big)}
\zeta^{-\Theta_\infty\sigma_3}, \quad z \to\infty,$$ when
$$(1+s_2s_3)e^{2\mathrm{i}\pi\Theta_\infty}+[s_1s_4+(1+s_3s_4)(1+s_1s_2)]e^{-2\mathrm{i}\pi\Theta_\infty}=2\cos 2\pi\Theta.$$
In our case we get $\Theta=c/2$, $\Theta_\infty=-c/2$.  It means
that, using the same strategy to Section \ref{sec4} and \ref{sec5},
we could get the similar result about the asymptotics of orthogonal
polynomial in terms of Painlev\'e IV equation:
$$\frac{d^2u}{dx^2}=\frac{1}{2u}\left(\frac{du}{dx}\right)^2+\frac{3}{2}u^3+4xu^2+(2+2x^2-4\Theta_\infty)u-\frac{8\Theta^2}{u},$$
where the solution $u$ is related to the Riemann-Hilbert problem by
$$u(x)=-2x-\frac{d}{dx}\log\big((\Psi_{-1})(x)_{12}\big).$$

\appendix
\section{Proof of Lemma \ref{lemma4} }\label{app2}
By \cite{Ha 2010} we can write
$$D_{-c}(\zeta)=2^{-c/2}e^{-\zeta^2/4}U\left(\frac{c}{2},\frac{1}{2},\frac{\zeta^2}{2}\right),$$
where $U$ has the following asymptotic expansion as
$|\zeta|\to\infty$.
\begin{equation}\nonumber
U\left(\frac{c}{2},\frac{1}{2},\frac{\zeta^2}{2}\right)=\left(\frac{\zeta^2}{2}\right)^{-\frac{c}{2}}\sum_{s=0}^{n-1}\left(-\frac{\zeta^2}{2}\right)^{-s}
\frac{\left(\frac{c}{2}\right)_{s}\left(\frac{c+1}{2}\right)_{s}}{s!(2\zeta^2)^s}+\widehat{\varepsilon_n}\left(\frac{\zeta^2}{2}\right),\quad|\rm
arg\,\zeta|< \frac{\pi}{2}.
\end{equation}
The error term $\widehat{\varepsilon_n}$ is bounded by
$$
\left|\widehat{\varepsilon_n}\left(\frac{\zeta^2}{2}\right)\right|\leq2^{\frac{c}{2}+n+1}\alpha\left|\frac{(\frac{c}{2})_n(\frac{c+1}{2})_n}
{n!(\zeta^2)^{n+\frac{c}{2}}}\right| \exp\left(\frac{4\alpha\rho
}{|\zeta^2|}\right),$$ where
$$\alpha=\frac{1}{1-\sigma},\quad
\sigma=\left|\frac{1-2c}{\zeta^2}\right|,\quad
\rho=\left|\frac{c^2-c+1}{4}\right|+\frac{\sigma(1+\frac{\sigma}{4})}{(1-\sigma)^2}.$$
We have
\begin{equation}\nonumber
|\varepsilon_n(\zeta)|=
2^{-\frac{c}{2}}|\zeta|^c\left|\widehat{\varepsilon_n}\left(\frac{\zeta^2}{2}\right)\right|\leq
C\left|\frac{(\frac{c}{2})_n(\frac{c+1}{2})_n}
{n!(\zeta^2)^{n}}\right|.
\end{equation} where
$$C=\frac{2^{n+1}|\zeta^2|}{(|\zeta^2|-|1-2c|)}
\exp\left(\left|\frac{c^2-c+1}{4(|\zeta^2|-|1-2c|)}\right|+\frac{|1-2c|(|\zeta^2|+\frac{|1-2c|}{4})}{(|\zeta^2|-|1-2c|)^3}\right).$$
For $|\zeta^2|/|1-2c|$ big enough, we have $C\leq 2^{n+2}$.

\section{Lax pair: how the numerical calculation is done}
Define $\widetilde{Y}(z)$ by
$\widetilde{Y}(z)=\widetilde{Y}_{n}(z)=Y(z)\begin{bmatrix}
\left(\frac{z-a}{z}\right)^c\frac{1}{e^{Naz}}&0 \\
0&z^{n}
\end{bmatrix},$ then the Riemann-Hilbert problem for $\widetilde{Y}(z)$ is
\begin{equation}\nonumber
    \begin{array}{ccc}
\left\{
\begin{array}{lll}
\widetilde{Y}(z) \text{ is holomorphic in }\mathbb{C}\setminus
\Gamma,\vspace{0.2cm}\\
\widetilde{Y}_+(z)=\widetilde{Y}_-(z)\begin{bmatrix}
1&1 \\
0&1
\end{bmatrix} ,&z\in\Gamma,\vspace{0.2cm}\\
\widetilde{Y}_+(z)=\widetilde{Y}_-(z)\begin{bmatrix}
e^{{2c\pi\mathrm{i}}}&0 \\
0&1
\end{bmatrix},&z\in(0,a),\vspace{0.2cm}\\
\widetilde{Y}(z)=\displaystyle\left(I+\mathcal
{O}\left(\frac{1}{z}\right)\right)\begin{bmatrix}
\left(\frac{z-a}{z}\right)^c\frac{z^n}{e^{Naz}}&0 \\
0&1
\end{bmatrix},& z\to\infty.
\end{array}
\right.
\end{array}
\end{equation}
We observe $\widetilde{Y}_n(z)$ and $\widetilde{Y}_{n+1}(z)$ have
the same jump matrices. Since $\det{Y}(z)\equiv 1$, the inverse of
${\widetilde{Y}(z)}$ exists in
$\mathbb{C}\setminus(\Gamma\cup(0,a))$, and we can define
$$A_n(z)=\frac{d\widetilde{Y}_n(z)}{dz}{\widetilde{Y}_n(z)}^{-1}.$$
The matrix function $A_n(z)$ is meromorphic and can be determined by
identifying the singularities. For $z\to\infty,$ writing (we know
that $c_n$ below is not related to the charge ``$c$'' in the
potential)
$$\widetilde{Y}_n(z)=\left(I+\frac{1}{z}\begin{bmatrix}
a_n&b_n \\
c_n&d_n
\end{bmatrix}+\cdots\right)\begin{bmatrix}
\left(\frac{z-a}{z}\right)^c\frac{z^n}{e^{Naz}}&0 \\
0&1
\end{bmatrix}, $$
we get
\begin{equation}\nonumber
A_n(z)=\begin{bmatrix}
-Na&0 \\
0&0
\end{bmatrix}+\frac{1}{z}\begin{bmatrix}
n&N a b_n \\
-N a c_n&0
\end{bmatrix}+\mathcal
{O}\left(z^{-2}\right).
\end{equation}

Similarly we get the following for $z\to 0$.
\begin{equation}\nonumber
\begin{array}{lll}
\widetilde{Y}_n(z)&=&\begin{bmatrix}
\alpha_n&\beta_n \\
\gamma_n&\eta_n
\end{bmatrix}\left(I+\mathcal
{O}(z)\right)\begin{bmatrix}
\left(\frac{z-a}{z}\right)^c\frac{1}{e^{Naz}}&0 \\
0&z^{n}
\end{bmatrix},\vspace{0.2cm}\\
A_n(z)&=&\frac{1}{z}\begin{bmatrix}
-c-(c+n)\beta_n\gamma_n&(c+n)\alpha_n\beta_n \\
-(c+n)\gamma_n\eta_n&n+(c+n)\beta_n\gamma_n
\end{bmatrix}.
\end{array}
\end{equation}
Therefore, we obtain
\begin{equation}\nonumber
\begin{split}
A_n(z)&=\begin{bmatrix}
-Na&0 \\
0&0
\end{bmatrix}+\frac{1}{z}\begin{bmatrix}
-c-(c+n)\beta_n\gamma_n&(c+n)\alpha_n\beta_n \\
-(c+n)\gamma_n\eta_n&n+(c+n)\beta_n\gamma_n
\end{bmatrix}\\
& +\frac{1}{z-a}\begin{bmatrix}
(c+n)\left(1+\beta_n\gamma_n\right)&Nab_n-(c+n)\alpha_n\beta_n \\
-Nac_n+(c+n)\gamma_n\eta_n&-n-(c+n)\beta_n\gamma_n
\end{bmatrix}.
\end{split}
\end{equation}
Defining $M_n(z)=\widetilde{Y}_{n+1}(z){\widetilde{Y}_n(z)}^{-1}$
 we obtain, by the similar procedure as above,
\begin{equation}\nonumber
M_n(z)=\begin{bmatrix}
z+a_{n+1}-a_n&-b_n \\
c_{n+1}&1
\end{bmatrix}.
\end{equation}
The compatibility of the Lax pair,
\begin{equation}\nonumber
\begin{split}
\frac{d\widetilde{Y}_n(z)}{dz}&=A_n(z){\widetilde{Y}_n(z)},\vspace{0.2cm}\\
\widetilde{Y}_{n+1}(z)&=M_n(z){\widetilde{Y}_n(z)},\end{split}
\end{equation}
gives
\begin{equation}\nonumber
A_{n+1}(z)M_n(z)=\frac{dM_n(z)}{dz}+M_n(z)A_n(z).
\end{equation}
This gives the following recurrence relation:
\begin{equation}\nonumber
\begin{split}
&a_{n+1}=a_n+\frac{b_n \left(1+\beta_n\gamma_n\right)}{\alpha_n
\beta_n},\quad\alpha_{n+1}=\frac{b_n}{\beta_n},\quad
\gamma_{n+1}=-\frac{1}{\beta_n},
\vspace{0.2cm}\\
 &b_{n+1}=\frac{(1+n+a^2N)b_n}{aN}+\frac{(c+n)\alpha_n\beta_n}{N}+\frac{b^2_n\left(1+\beta_n\gamma_n\right)}
{\alpha_n\beta_n},\vspace{0.2cm}
\\
&\beta_{n+1}=\frac{\tilde{c}}{(1+c+n)
\left((c+n)\alpha_n\beta_n-aNb_n\right)\alpha^2_n\beta_n},
\end{split}
\end{equation}
where
\begin{equation}\nonumber
\begin{split}
\tilde{c}&=a^2N-c-a(1+2(c+n))\alpha_n\beta_n
+\left(a^2N-c-a(c+n)\alpha_n\beta_n\right)\beta_n\gamma_n\\
&+(c+n)(c+n+1)\alpha^3_n\beta^3_n+aN^2b^3_n
\left(1+\beta_n\gamma_n\right)^2,\\
&a_0=0,\quad
b_0=a,\quad\alpha_0=1,\quad\beta_0=1+a^2N,\quad\gamma_0=0.
\end{split}
\end{equation}
The last line contains the initial condition of the recurrence
relation. We used the above relation to generate the orthogonal
polynomials numerically.

\end{document}